\documentclass{article}
\usepackage{amsfonts,amsmath,amsthm,amssymb}
\usepackage[numbers,sort]{natbib}
\usepackage{geometry}
\usepackage{amsmath}

\DeclareFontFamily{U}{mathx}{}
\DeclareFontShape{U}{mathx}{m}{n}{<-> mathx10}{}
\DeclareSymbolFont{mathx}{U}{mathx}{m}{n}
\DeclareMathAccent{\widecheck}{0}{mathx}{"71}

\usepackage{enumitem}

\usepackage{bm}

\usepackage{xcolor}
\RequirePackage[colorlinks,citecolor=blue,urlcolor=blue]{hyperref}

\theoremstyle{plain}
\newtheorem{lem}{Lemma}
\newtheorem{cor}[lem]{Corollary}
\newtheorem{prop}[lem]{Proposition}
\newtheorem{thm}[lem]{Theorem}

\theoremstyle{definition}
\newtheorem{remark}[lem]{Remark}
\newtheorem{defn}[lem]{Definition}

\newtheorem{assum}[lem]{Assumption}

\newcommand{\R}{\mathbb{R}}
\renewcommand{\P}{\mathbb{P}}

\newcommand{\E}{\mathbb{E}}
\newcommand{\F}{\mathcal{F}}
\newcommand{\N}{\mathbb{N}}

\numberwithin{equation}{section}
\numberwithin{lem}{section}

\renewcommand{\S}{{\mathbb S}}

\newcommand{\Bcal}{{\mathcal B}}
\newcommand{\Ccal}{{\mathcal C}}
\newcommand{\Dcal}{{\mathcal D}}
\newcommand{\Ecal}{{\mathcal E}}
\newcommand{\Fcal}{{\mathcal F}}

\newcommand{\Wcal}{{\mathcal W}}

\DeclareMathOperator{\Tr}{\mathop{tr}}

\let\div\relax
\DeclareMathOperator{\div}{\mathop{div}}

\makeatletter
\renewenvironment{proof}[1][\proofname] {\par\pushQED{\qed}\normalfont\topsep6\p@\@plus6\p@\relax\trivlist\item[\hskip\labelsep\bfseries#1\@addpunct{.}]\ignorespaces}{\popQED\endtrivlist\@endpefalse}
\makeatother

\begin{document}
	
	\title{Ergodic robust maximization of asymptotic growth with stochastic factor processes\footnote{We would like to thank two anonymous referees for valuable comments, in particular for suggesting the problem formulation of Section~\ref{sec:main_extended} involving a pre-specified covariation structure between the price process and factor process.
}}
	\author{David Itkin\footnote{Department of Statistics, London School of Economics, \url{d.itkin@lse.ac.uk}}\and Benedikt Koch\footnote{Department of Statistics, Harvard University, \url{ benedikt_koch@g.harvard.edu}} \and Martin Larsson\footnote{Department of Mathematical Sciences, Carnegie Mellon University, \url{larsson@cmu.edu}. This work has been partially supported by the National Science Foundation under grant
NSF DMS-2206062} \and Josef Teichmann\footnote{Department of Mathematics, ETH Z\"urich, \url{josef.teichmann@math.ethz.ch}}}
	
		\maketitle

	\begin{abstract}
	We consider a robust asymptotic growth problem under model uncertainty in the presence of stochastic factors. We fix two inputs representing the instantaneous covariance for the asset price process $X$, which depends on an additional stochastic factor process $Y$, as well as the invariant density of $X$ together with $Y$. The stochastic factor process $Y$ has continuous trajectories but is not even required to be a semimartingale. Our setup allows for drift uncertainty in $X$ and model uncertainty for the local dynamics of $Y$. This work builds upon a recent paper of Kardaras \& Robertson \cite{kardaras2021ergodic}, where the authors consider an analogous problem, however, without the additional stochastic factor process. Under suitable, quite weak assumptions we are able to characterize the robust optimal trading strategy and the robust optimal growth rate. The optimal strategy is shown to be functionally generated and, remarkably, does not depend on the factor process $Y$. Our result provides a comprehensive answer to a question proposed by Fernholz in 2002. We also show that the optimal strategy remains optimal even in the more restricted case where $Y$ is a semimartingale and the joint covariation structure of $X$ and $Y$ is prescribed as a function of $X$ and $Y$. Our results are obtained using a combination of techniques from partial differential equations, calculus of variations, and generalized Dirichlet forms.
	\end{abstract}
	
	
	
	\section{Introduction}
	In this paper we consider an investor's asymptotic growth maximization problem under model uncertainty, in the presence of stability (i.e.~ergodicity of the market dynamics) and stochastic factors affecting, e.g., volatilities and covariances. We study the quantity
	\[\sup_{\theta \in \Theta} \inf_{\P \in \Pi} g(\theta;\P),\]
	where $\Theta$ is the set of all admissible trading strategies (adapted to some filtration with respect to which $X$ is adapted), $\Pi$ is the set of admissible models and $g(\theta;\P)$ is the investor's asymptotic growth rate when they employ the strategy $\theta$ and the market dynamics is governed by the law $\P \in \Pi$. The market consists of $d$ risky assets with discounted price process $X = (X^1,\ldots,X^d)$ taking values in a set $E \subset \R^d$, as well as a risk-free asset which serves as numeraire (we shall never use it and therefore do not need to introduce it). In addition, there is a possibly observed factor process $Y = (Y^1,\ldots,Y^m)$ with values in a set $D \subset \R^m$, which can influence the dynamics of $X$ but is \emph{not} itself traded. Strategies $\theta$ do in general depend on past values of $X$ and possibly on the past values of $Y$ depending on our informational assumptions, and a model $\P$ is a joint law on path space of $(X,Y)$. Precise definitions and assumptions are given in Section~\ref{sec:setup}. Our work builds on that of \cite{kardaras2021ergodic}, which considers a similar robust growth problem in the presence of stability, but does not allow for stochastic factors.
	
	The main motivation for this line of research comes from the empirically observed stability of the \emph{ranked market weights} in US equity markets, which can be related to the existence of an invariant measure; see \cite[Chapter~5]{fernholz2002stochastic}. In this setting the numeraire is the market's total capitalization. Stochastic Portfolio Theory (SPT) is a framework designed to study this stability and exploit it for investment problems. In our setting this corresponds to $E$ being the $d$-dimensional open simplex and viewing $X$ as the market weight process. Many researchers have proposed models that produce stable market weight processes, investigated approaches to statistically estimating the capital distribution curve and studied the performance of trading strategies under this observed stability \cite{itkin2021open,ichiba2011hybrid,banner2005atlas,fernholz2005relative,campbell2021functional,cuchiero2019polynomial,ruf2019general,  cuchiero2019cover,ruf2020impact}. 
	
	Fernholz, in his celebrated book on SPT, proposes the research question of developing a theory of portfolio optimization that depends only on observable quantities \cite[Problems~3.1.7 \& 3.1.8]{fernholz2002stochastic}. Observable quantities are prices, maybe instantaneous covariance along the trajectory as well as some observable factors, which influence it, and some invariant laws. Neither instantaneous drifts (nor jump structures, which we do not consider here) are pathwise observable due to the level of noise typical of financial markets. One class of portfolios solely depending on those observables, for which performance guarantees can be deduced under model uncertainty, are \emph{functionally generated portfolios}. Indeed, an investor who trades using a functionally generated portfolio can compute their prescribed holdings by only using current market prices.  Moreover, performance guarantees for functionally generated portfolios can be obtained assuming some knowledge on instantaneous covariance and stability of the asset process \cite{fernholz2002stochastic,karatzas2017trading,itkin2020robust,larsson2021relative,fernholz2005relative}.
	
	The work of Kardaras \& Robertson \cite{kardaras2021ergodic} makes important progress towards solving the question introduced by Fernholz, in particular what is the particular role of functionally generated portfolios. They take two inputs:
	\begin{enumerate}[noitemsep]
	\item a matrix valued function $\bar c_X(x)$, which serves as the instantaneous covariance matrix for $X$,
	\item a positive function $\bar p(x)$ with $\int_E \bar p = 1$ serving as the invariant density for $X$, encoding stability.
	\end{enumerate}
	They consider the class of admissible models $\Pi$ to consist of all laws under which
	\begin{enumerate}[noitemsep]
	    \item \label{item:KR1} $X$ is a continuous semimartingale with $\bar c_X(X_t)$ as its instantaneous covariance matrix,
	    \item \label{item:KR2} $\lim_{T \to \infty} \frac{1}{T} \int_0^T h(X_t)\, dt = \int_E h\bar p$ for any $h \in L^1(\bar p)$,
	    \item \label{item:KR3} the family of laws of $X_t$, $t \geq 0$, is tight. 
	\end{enumerate} This class includes all measures with the specified covariance structure and specified stable behaviour, but the setup allows of course for some drift uncertainty. The more assets we have the richer is the class of admissible models $\Pi$. The authors are then able to show that the robust optimal strategy is functionally generated and obtain a partial differential equation (PDE) characterization for the optimal generating function. They do this in a two-step procedure. First, the asymptotic growth rate is optimized over a class of functionally generated portfolios. For this class of strategies performance guarantees can be derived under any admissible measure. Secondly, an admissible worst-case measure $\widehat \P$ is constructed under which the functionally generated portfolio from the first step is growth-optimal over \emph{all} portfolios. With those two ingredients one can fully solve the min-max problem. This result is surprising in two ways: first general utility optimization problems have optimal strategies with at least some path dependence via the wealth process (which is a state variable), and, second, it is a priori not at all clear that there is a worst case model in our class $\Pi$, where the functionally generated portfolio is overall optimal. 
	
    One striking limitation of the general setup in \cite{kardaras2021ergodic} is, however, that the instantaneous covariance matrix $c_X$ is in feedback form, i.e.~does not depend on further factors. Indeed, it is widely accepted that the volatility of asset price processes is influenced by market forces other than the current price level (see e.g.\ \cite{engle2007good}). To address this issue, we model the instantaneous covariance of $X$ via $c_X(X_t,Y_t)$, where $Y$ is a stochastic factor process as mentioned above and, similarly to \cite{kardaras2021ergodic}, $c_X(x,y)$ is a fixed matrix valued function taken as input. One can interpret $Y$ as a factor with econometric interpretation driving stochastic covariance, or as a factor modelling uncertainty of the choice of $x \mapsto c_X(x,.)$. In the main part of the paper we do not prescribe the dynamics of $Y$, but we do impose the restriction that $(X,Y)$ is jointly ergodic with a given unique limiting density $p(x,y)$. In particular, $Y$ is not required to be a semimartingale, but its ergodicity and limiting density (jointly with $X$) are known, a common assumption in econometrics. We are again able to show that the robust optimal portfolio is functionally generated, obtain a PDE characterization for it and, remarkably, show that its portfolio weights do not depend on the factor process $Y$ but only on the current prices $X$. This result provides a more complete answer to Fernholz's questions \cite[Problems~3.1.7 \& 3.1.8]{fernholz2002stochastic} by greatly generalizing the admissible class of measures to allow for stochastic covariance. Theorem~\ref{thm:main} in Section~\ref{sec:main_result} below contains the precise statement of this result.
    
    In Section~\ref{sec:main_extended} we consider the analogous problem where $Y$ is additionally assumed to be a semimartingale and the entire joint covariation matrix $c(x,y)$ of $X$ and $Y$ is additionally prescribed as an input. Remarkably, the same strategy remains robust growth-optimal in this setting. In particular, the strategy is again functionally generated, does not depend on $Y$ and additionally does not depend on the form of the off-diagonal and lower-right blocks $c_{XY}(x,y)$ and $c_Y(x,y)$ of the given matrix $c(x,y)$, which determine the behavior of $\langle X, Y \rangle$ and $\langle Y \rangle$. This lack of dependence may appear counter-intuitive or at least surprising; the reasons for it are discussed in Section~\ref{sec:discussion_Y}.
    
    From a mathematical point of view, the introduction of a stochastic covariance process introduces a delicate technical challenge. Indeed, while the construction of the worst-case measure $\widehat \P$ and verification of the ergodic property in \cite{kardaras2021ergodic} were done using the classical theory of positive harmonic functions (see \cite{pinsky1995positive} for a reference), in our case these results are no longer applicable due to a loss in regularity of the coefficients induced by the inclusion of the stochastic covariance factor process.
    
    To overcome this difficulty we modify the input matrix $c_X$ on an arbitrarily small neighbourhood near the boundary of $E$ (although in the absence of a factor process, $c_X$ is actually left unchanged). Then using a combination of ideas from elliptic PDE theory, calculus of variations and recent developments in the theory of generalized Dirichlet forms \cite{lee2020analyt,gim2018recurrence} we are able to construct an admissible worst-case measure $\widehat \P$ (see Theorem~\ref{thm:worst_case_measure_existence} below). Importantly, the robust optimal strategy and robust optimal growth rate only depend on the original inputs $c_X(x,y)$ and $p(x,y)$, not on the modified inputs. Additionally, we show in Corollary~\ref{cor:original} that in certain special cases the modifications are unnecessary. As such, we view these modifications as purely technical in nature. This point is further discussed in Section~\ref{sec:K-mod}.
    
    The layout of the paper is as follows. Section~\ref{sec:setup} introduces the mathematical framework and formulates the robust asymptotic growth problem we study. Section~\ref{sec:heuristics} contains a heuristic discussion of the main ideas to tackle the problem. Technical assumptions are introduced in Section~\ref{sec:assumptions}. Section~\ref{sec:results} contains the precise statements of all our results with Section~\ref{sec:main_result} considering the  problem with inputs $c_X$ and $p$, while Section~\ref{sec:main_extended} considers the extended problem where the entire joint covariation matrix $c$ is an input. Section~\ref{sec:discussion} contains a discussion of the results emphasizing the financial interpretation and relationship to the results in \cite{kardaras2021ergodic}. We consider several pertinent examples in Section~\ref{sec:examples}. Appendix~\ref{sec:PDE} contains results for a class of degenerate elliptic PDEs corresponding to certain variational problems, and Appendix~\ref{S_measurability} contains a measurability result for parameter dependent versions of such problems. We were unable to find a standard reference for these results and so have developed them in a general framework in the appendix. Lastly, Appendix~\ref{sec:proofs} contains the proofs of most of the results stated in Section~\ref{sec:results}. In particular, the results developed in Appendix~\ref{sec:PDE} are crucially applied in Appendix~\ref{sec:proofs} to construct the worst-case measure $\widehat \P$ described above.

\section{Setup} \label{sec:setup}

We fix integers $d,m \geq 1$ and non-empty open sets $E \subset \R^d$, $D \subset \R^m$. Set $F = E \times D$. We generically denote elements of $E$ by $x$, elements of $D$ by $y$, and write $z =(x,y)$ for elements of $F$. For a function $u$ we write $\nabla_xu$ for $(\partial_{x_1} u,\dots,\partial_{x_d}u)$, $\div_x u$ for $\sum_{i=1}^d \partial_{x_i} u$, and use $\nabla_yu$ and $\div_y u$ analogously.
	
The set $E$ serves as the state space for a $d$-dimensional asset price process $X$, while $D$ serves as the state space for an $m$-dimensional non-traded factor process $Y$. Consequently the joint $(d+m)$-dimensional process $Z =(X,Y)$ has state space $F$. We realize $Z$ as the coordinate process on the canonical space $(\Omega,\Fcal,\mathbb F)$, where $\Omega$ is the space of all $F$-valued continuous trajectories with the topology of locally uniform convergence, $\Fcal$ is the Borel $\sigma$-algebra, and $\mathbb{F} = (\F_t)_{t \geq 0}$ is the right-continuous filtration generated by $Z$. Since $X$ serves as the asset price process it will be a continuous semimartingale under all laws on $\Omega$ considered below. The factor process $Y$, however, will not necessarily be a semimartingale.

Trading strategies are modeled by $d$-dimensional predictable processes $\theta$. The investor's wealth process, assuming unit initial capital, is
\[V^\theta = \Ecal\left(\int_0^\cdot \theta_t^\top dX_t\right),\]
where $\Ecal(\cdot)$ denotes the Dol\'eans-Dade exponential. This is well-defined up to nullsets under any law $\P$ such that $X$ is a semimartingale and $\theta$ is $X$-integrable. Notice that the nullsets will depend on $\P$, that it is by no means clear that $\Pi$ is dominated, and that we are only considering strategies $\theta$ which are integrable with respect to $X$ for \emph{every} $\P$. The goal is to maximize the investor's asymptotic growth rate, which is defined as follows.

\begin{defn}[Asymptotic growth rate]
For any law $\P$ on $\Omega$ such that $X$ is a semimartingale, and any predictable $X$-integrable process $\theta$, the \emph{asymptotic growth rate} is
\[
g(\theta;\P) = \sup\left\{\gamma \in \R: \lim_{T\to \infty} \P(T^{-1}\log V^\theta_T \geq \gamma) = 1 \right\}.
\]
For a collection $\Pi$ of such laws, the \emph{robust asymptotic growth rate} of a trading strategy $\theta$, required to be $X$-integrable under every $\P \in \Pi$, is
\[
\inf_{\P \in \Pi} g(\theta;\P).
\]
\end{defn}
The robust asymptotic growth rate of $\theta$ is the worst-case rate achieved by $\theta$ across all market models in $\Pi$. The robust asymptotic growth problem is to maximize this worst-case rate. Thus the goal, akin to the one in \cite{kardaras2021ergodic}, is to study the quantity
\begin{equation} \label{eqn:lambda_def}
	\lambda_\Pi := \sup_{\theta \in \Theta}\inf_{\P \in \Pi} g(\theta;\P),
\end{equation}
where $\Theta$ is the set of all $d$-dimensional predictable processes that are $X$-integrable under every $\P \in \Pi$. Here the $\Pi$-dependence of $\Theta$ is suppressed from the notation; note however that regardless of $\Pi$, $\Theta$ always contains all predictable and locally bounded strategies.

The solution to the robust asymptotic growth problem depends on the choice of $\Pi$ in general. We now describe the principal choices of $\Pi$ appearing in our work. We take as input two functions
 \begin{equation} \label{eqn:input_functions}
 c_X:F \to \mathbb{S}^d_{++} \qquad \text{and} \qquad p:F \to (0,\infty),
 \end{equation}
 where $\mathbb{S}^d_{++}$ is the set of symmetric positive definite $d \times d$ matrices. Below, $c_X$ serves as the instantaneous covariance matrix for $X$ (which is a function of both $X$ and $Y$), and $p$ as the joint invariant density of $X$ and $Y$. We impose the following regularity assumptions on these inputs.

 \begin{assum}[Regularity assumptions]  \label{ass:regularity} 
		 For a fixed $\gamma \in (0,1]$,
		 \begin{enumerate}[noitemsep]
			\item  $D$ is a bounded convex set,
			\item \label{item:regular_2} $c_X \in W^{1,\infty}_{\mathrm{loc}}(F)$ and for every $y \in D$, $c_X(\cdot,y) \in C^{2,\gamma}(E)$,
			\item \label{item:regular_3} $p \in W^{2,\infty}_{\mathrm{loc}}(F)$ and for every $y \in D$, $p(\cdot,y) \in C^{2,\gamma}(E)$. Additionally, $\int_F p  = 1$.
		\end{enumerate}
\end{assum}

Here $W^{k,p}_{\mathrm{loc}}(F)$ is the Sobolev space of $k$-times weakly differentiable functions on $F$ whose weak derivatives up to order $k$ (including the function itself) belong to $L^p(U)$ for every set $U$ with compact closure in $F$. By Sobolev embedding, any element of $W^{k,\infty}_{\mathrm{loc}}(F)$ has a continuous version for $k \geq 1$. We always use this version, and note that continuity need not hold up to the boundary of $F$. $C^{2,\gamma}(F)$ is the set of twice differentiable functions whose second derivatives are $\gamma$-H\"older continuous. In each case the co-domain is understood from the context.

We may now define our first class of models.

\begin{defn}[First admissible class of measures] \label{def:upsilon}
	Let Assumption~\ref{ass:regularity} be satisfied. We denote by $\Pi_0$ the set of all laws $\P$ on $\Omega$ such that the following conditions hold:
	\begin{enumerate}
		\item \label{item:qv-0} $X$ is a $\P$-semimartingale with covariation process $\langle X \rangle = \int_0^\cdot c_X(Z_t)\, dt$,
		\item  \label{item:ergodic} for any locally bounded $h \in L^1(F,\mu)$, where $ \mu(dz):=p(z)dz $,
		\[\lim_{T \to \infty} \frac{1}{T}\int_0^T h(Z_t)\, dt = \int_F hp; \quad \P\text{-a.s.}, \]
		\item \label{item:tight} the family of laws under $\P$ of $X_t$, $t \geq 0$, is tight. 
	\end{enumerate}
\end{defn}

In Section~\ref{sec:results} we solve the robust asymptotic growth problem for the collection $\Pi_0$ under certain fairly implicit additional assumptions on the inputs $c_X$ and $p$; see Corollary~\ref{cor:original}. To allow for weaker and more direct assumptions we consider slightly larger collections $\Pi_K$, which we now introduce.

Given $c_X$ and $p$ as in Assumption~\ref{ass:regularity}, define the averaged instantaneous covariance function $A:E \to \mathbb{S}^d_{++}$ whose components are
\begin{equation} \label{eqn:A_def}
    A^{ij}(x) = \int_D c_X^{ij}(x,y)p(x,y)\, dy; \qquad i,j=1,\dots,d, \quad x \in E.
\end{equation}
We then consider those covariance functions $\tilde c_X : F \to \S^d_{++}$ which, like $c_X$, have $A$ as their average, and additionally coincide with $c_X$ on a compact set. More precisely, for $c_X$ and $p$ as in Assumption~\ref{ass:regularity} and a compact set $K \subset E$ (possibly empty), we define
\[
\Ccal_K = \left\{ \tilde c_X \in W^{1,\infty}_{\mathrm{loc}}(F): \tilde c_X = c_X \text{ on } K \times D \text{ and } \int_D \tilde c_X(x,y)p(x,y)\, dy = A(x) \text{ for } x \in E \right\}.
\]
We refer to an element $\tilde c_X$ of $\Ccal_K$ as a \emph{$K$-modification} of $c_X$. 

\begin{remark}\label{rem:K-empty}
As mentioned above, $K=\emptyset$ is allowed. The first condition in the definition of $\Ccal_\emptyset$ is then vacuous, and the only remaining requirement is that $\tilde c_X$ averages to $A$.
\end{remark}

\begin{defn}[Second admissible class of measures] \label{def:upsilon-2}
	Let Assumption~\ref{ass:regularity} be satisfied. For a compact set $K \subset E$ we denote by $\Pi_K$ the set of all laws $\P$ on $\Omega$ such that \ref{item:ergodic} and \ref{item:tight} of Definition~\ref{def:upsilon} hold, along with the modified condition
	\begin{enumerate}[label={\rm(\roman{enumi}')}]
		\item \label{item:qv} $X$ is a $\P$-semimartingale with covariation process $\langle X \rangle = \int_0^\cdot \tilde c_X(Z_t)\, dt$ for some $\tilde c_X \in \Ccal_K$.
	\end{enumerate}
\end{defn}

Note that $\Pi_0 \subset \Pi_K$ for any compact set $K \subset E$, and that $\Pi_0 = \bigcap_K \Pi_K$ where the intersection extends over all such compact sets. Additionally, if $K_1 \subset K_2$ then $\Pi_{K_2} \subset \Pi_{K_1}$ so that $\Pi_\emptyset$ is the largest class of measures we consider in this paper. Our main result, Theorem~\ref{thm:main}, solves the optimal robust growth problem for the enlarged collections $\Pi_K$. Remarkably, it turns out that the solution does not depend on the set $K$, and coincides with the aforementioned solution (obtained under more stringent assumptions) to the robust problem for $\Pi_0$ that is the content of Corollary~\ref{cor:original}. Furthermore, we argue in Section~\ref{sec:K-mod} that $\Pi_0$ and $\Pi_K$ become statistically indistinguishable when $K$ is chosen sufficiently large. 

\section{Heuristic argument} \label{sec:heuristics}
We first observe that there is a natural class of strategies $\theta$ that posses the \emph{growth rate invariance} property that
\[
g(\theta;\P)  \text{ is independent of } \P \in \Pi_\emptyset.
\]
Indeed, this is the case for the class
\[
\Theta_0 = \{\nabla\phi(X): \phi \in \Dcal\},
\]
where
\begin{equation} \label{eqn:Dcal} \Dcal = \left\{\phi \in C^2(E): \int_E \left|\frac{\Tr(A\, \nabla^2 e^{\phi})}{e^{\phi}}\right| \, < \infty \right\}.
\end{equation}
The strategies $\theta \in  \Theta_0$ are of the gradient form $\theta = \nabla\phi(X)$ and are known as \emph{functionally generated portfolios}. To see that any such strategy has the growth rate invariance property, apply It\^o's formula under any measure $\P \in \Pi_\emptyset$ to get
\[\log V_T^{\theta} = \phi(X_T) - \phi(X_0) - \frac{1}{2}\int_0^T \frac{\Tr(\tilde c_X(X_t,Y_t)\nabla^2 e^{\phi(X_t)})}{e^{\phi(X_t)}}\, dt,\]
where $\tilde c_X \in \Ccal_\emptyset$ is the covariance matrix of $X$ under $\P$; see Definition~\ref{def:upsilon-2}.
By tightness of the laws of $X_T$, $T\geq 0$, we have that $\phi(X_T)/T \to 0$ in probability as $T \to \infty$. Hence, by the ergodic property and the definition of $\tilde c_X \in \Ccal_\emptyset$ (see also Remark~\ref{rem:K-empty}) it follows that
\begin{equation} \label{eqn:c2_growth_rate}
    g(\theta;\P) = -\frac{1}{2}\int_F \frac{\Tr(\tilde c_X(x,y)\nabla^2 e^{\phi(x)})}{e^{\phi(x)}}p(x,y) \, dx\, dy = - \frac{1}{2}\int_E \frac{\Tr(A(x)\nabla^2 e^{\phi(x)})}{e^{\phi(x)}}\, dx.
\end{equation}

We maximize the right hand side of \eqref{eqn:c2_growth_rate} over functions $\phi \in \Dcal$. This is a non-trivial procedure due to poor compactness properties of $\Dcal$, but as in \cite{kardaras2021ergodic} we are able to show that a maximizer $\hat \phi$ exists. Let $\hat \theta  = \nabla \hat \phi(X)$ be the associated trading strategy. By the growth rate invariance property and optimality of $\hat \phi$ this yields a \emph{lower bound} on the robust growth rate $\lambda_{\Pi_K}$, as defined in \eqref{eqn:lambda_def}, associated with $\Pi_K$ for any compact set $K \subset E$. Indeed,
\begin{equation} \label{eqn:lambda_lower_bound}
    \lambda_{\Pi_K} \geq \lambda_{\Pi_\emptyset} \ge \sup_{\theta \in \Theta_0}\inf_{\P \in \Pi_\emptyset} g(\theta;\P) = \sup_{\theta \in \Theta_0} g(\theta; \P^0) = g(\hat \theta;\P^0),
\end{equation}
where $\P^0$ is an arbitrary measure in $\Pi_K$.

To obtain an \emph{upper bound} we construct a measure $\widehat \P^K \in \Pi_K$ under which $\hat \theta$ is growth optimal among \emph{all} strategies $\theta \in \Theta$. The requirement that $\hat \theta$ be growth optimal, along with Definition~\ref{def:upsilon-2}\ref{item:qv}, pins down the required dynamics of $X$ under $\widehat \P^K$. Namely,
\begin{equation} \label{eqn:X_dynamics}
dX_t = \tilde c_X(X_t,Y_t)\nabla\hat \phi(X_t)\, dt + \tilde c_X^{1/2}(X_t,Y_t)\, dW_t^X,
\end{equation}
where $\tilde c_X \in \Ccal_K$,
$\tilde c_X^{1/2}$ is a matrix square root of $\tilde c_X$, and $W^X$ is a standard $d$-dimensional Brownian motion. The dynamics of the stochastic factor $Y$ are, at the moment, unspecified, but suppose for the time being that we have specified them in such a way that $\widehat \P^K \in \Pi_K$. Then we obtain 
\begin{equation} \label{eqn:lambda_upper_bound}
    \lambda_{\Pi_K} \leq \sup_{\theta \in \Theta} g(\theta; \widehat \P^K) = g(\hat \theta;\widehat \P^K).
\end{equation}
Thus, by taking $\P^0 = \widehat\P^K$ in \eqref{eqn:lambda_lower_bound}, it follows that equality must hold in \eqref{eqn:lambda_lower_bound} and in \eqref{eqn:lambda_upper_bound}. This characterizes $\lambda_{\Pi_K}$ and establishes $\hat \theta = \nabla \hat\phi(X)$ as the robust growth-optimal strategy. Notably, neither $\hat \theta$ nor $\lambda :=\lambda_{\Pi_K}$ depends on the choice of $K$.

We view $\widehat \P^K$ as a worst-case measure. This is because one cannot outperform the robust growth-optimal strategy $\hat \theta$ under this measure. The main difficulty in constructing $\widehat \P^K$ is to specify the dynamics of $Y$ such that the ergodic property, Definition~\ref{def:upsilon}\ref{item:ergodic}, holds. Our construction of $\widehat \P^K$ requires additional assumptions on the inputs $c_X$ and $p$, which are stated as Assumption~\ref{ass:nondegeneracy} in the next section. We note that our method ensures the existence of \emph{a} worst-case measure (i.e.\ a measure in $\Pi_K$ under which $X$ has dynamics \eqref{eqn:X_dynamics}), but the worst-case measure is not unique in general. Indeed, Example~\ref{ex:beta} exhibits a situation with uncountably many worst-case measures.

\section{Nondegeneracy assumptions} \label{sec:assumptions}

To carry out the program laid out in Section~\ref{sec:heuristics} we need to construct \emph{a} worst-case measure $\widehat \P^K \in \Pi_K$. Our construction is such that under $\widehat \P^K$
\begin{itemize}[noitemsep]
    \item $X$ has dynamics given by \eqref{eqn:X_dynamics},
    \item $Y$ will be a continuous semimartingale,
    \item $X$ will have zero covariation with $Y$; that is, $d\langle X,Y\rangle_t = 0$ for every $t$. 
\end{itemize}
We stress that $(X,Y)$ only possesses these very special properties under the worst-case measure $\widehat \P^K$ we construct.  These properties are not assumed to hold for arbitrary measures in the class $\Pi_K$, compare also Subsection \ref{sec:main_extended} where additional constraints are imposed. In particular, laws where $X$ and $Y$ have nontrivial covariation are permitted and, as previously mentioned, laws where $Y$ is not a semimartingale are permissible in $\Pi_K$. Our main result in this setting is given in Section~\ref{sec:main_result}. In Section~\ref{sec:main_extended} we discuss the related problem where the joint covariation structure of $X$ and $Y$ is prescribed and we explain how a compatible worst-case measure can be constructed. 

To carry out our construction, and to even establish that the class $\Pi_K$ is nonempty, we impose the following additional assumption. 

\begin{assum} [Nondegeneracy assumptions] \label{ass:nondegeneracy}
	We assume that there exists $c_Y:F \to \mathbb{S}^m_{++}$ such that the conditions below hold. To simplify the notation set $c = \text{diag}(c_X,c_Y):F \to \mathbb{S}^{d+m}_{++}$ and define $\ell= \frac{1}{2}c^{-1}\text{div}\,c + \frac{1}{2}\nabla \log p$ where $\text{div}\, c^i = \sum_{j=1}^{d+m} \partial_{j} c^{ij}$ for $i=1,\dots,d+m$. We canonically decompose $\ell$ as $\ell = (\ell_X,\ell_Y)$. We assume that
	\begin{enumerate}
		\item \label{item:c_y_sobolev} $c_Y\in W^{2,\infty}_{\mathrm{loc}}(F)$ and $c_Y(x,\cdot)p(x,\cdot)$ is locally Lipschitz continuous for every $x \in E$.
		\item \label{item:assum_new_4} For every fixed $x \in E$, there exists a constant $k_x \geq 0$ and a concave function $\rho_x:D \to (0,\infty)$ such that $\lambda_{\mathrm{min}}(c_Y(x,y))p(x,y) = \rho_x(y)^{k_x}$.
		\item  \label{item:assum_new_3}  For every fixed $x \in E$ and every $C \in \R$, $b \in \R^d$, $M \in \R^{d \times d}$ we have 
		\begin{equation} \label{eqn:integrability_assum}
		    \int_D\frac{ \div_x(c_X\ell_X)^2  + (\ell_X^\top c_X (\nabla_x \log p + b))^2  + \Tr(c_XM)^2 + C}{\lambda_{\mathrm{min}}(c_Y)}p < \infty.
		\end{equation}
		\item \label{item:technical_1} $\int_{F} \ell^\top c \ell p < \infty$.
		\item \label{item:div_int} $\int_F \div_x(c_X\ell_Xp)\, < \infty$.
\item \label{item:test_func}
There exist functions $\varphi_n \in C_c^\infty(E)$ and $\psi_n \in C_c^\infty(D)$ satisfying $0 \leq \varphi_n, \psi_n \leq 1$, $\lim_{n \to \infty} \varphi_n = \lim_{n \to \infty} \psi_n = 1$ and 
\begin{equation} \label{eqn:test_functions}
    \lim_{n \to \infty} \int_E \nabla \varphi_n^\top A \nabla \varphi_n(x)\, dx = \lim_{n \to \infty} \int_D \nabla \psi_n^\top B\, \nabla \psi_n(y)\, dy = 0,
\end{equation}
where $A$ is given by \eqref{eqn:A_def} and $B(y) = \int_E c_Y(x,y)p(x,y)\, dx$ is assumed to be finite for almost every $y \in D$.
	\end{enumerate} 	
\end{assum}

\begin{remark}
	A candidate choice is $c_Y(x,y) = h(x)\rho^k_{\partial D}(y)/ p(x,y)I_m$ for some $k \geq 0$, where $I_m$ is the $m \times m$ identity matrix, $h$ is a positive continuous function and $\rho_{\partial D}$ is a \emph{regularized distance to the boundary} of the convex domain $D$. The latter means that $\rho_{\partial D}$ is $C^2$, concave and such that there exists a universal constant $C > 0$ with $\frac{1}{C}\mathrm{dist}(y,\partial D) \leq \rho_{\partial D}(y) \leq C\, \mathrm{dist}(y,\partial D)$ for every $y \in D$; see \cite[Theorem~1.4]{lieberman1985regular} for a construction of $\rho_{\partial D}$. In this case Assumption~\ref{ass:nondegeneracy}\ref{item:c_y_sobolev} is satisfied. Additionally,  $\lambda_{\min}(c_Y(x,y))p(x,y) = h(x)\rho^k_{\partial D}(y)$ so that Assumption~\ref{ass:nondegeneracy}\ref{item:assum_new_4} is also satisfied thanks to the concavity of $\rho_{\partial D}$. Hence, under this choice it just remains to check the integrability conditions Assumption~\ref{ass:nondegeneracy}\ref{item:assum_new_3}-\ref{item:test_func}. 
\end{remark}

Assumption~\ref{ass:nondegeneracy}\ref{item:assum_new_4} ensures that a certain weighted Poincar\'e inequality holds, which is crucial for our construction of a worst case measure $\widehat \P^K$. Condition \ref{item:assum_new_3} is used to verify that $\widehat \P^K$ satisfies the ergodic property. Conditions \ref{item:technical_1}-\ref{item:test_func} are analogues of \cite[Assumption~1.4]{kardaras2021ergodic} in our setting. The integrability bounds \ref{item:technical_1} and \ref{item:div_int} ensure that the robust optimal growth rate is finite and achieved by a strategy $\hat \theta \in \Theta$. Assumption~\ref{ass:nondegeneracy}\ref{item:test_func} is needed for well-posedness of the problem.  Indeed, in \cite{kardaras2021ergodic} it was shown that, in the one dimensional case, if \cite[Assumption~1.5(iii)]{kardaras2021ergodic} fails then the corresponding robust optimal growth problem of that paper is degenerate in the sense that either the admissible class of probability measures is empty or the robust growth rate is infinite. \cite[Assumption~1.5(iii)]{kardaras2021ergodic} can be equivalently rephrased in terms of test function conditions, similar to \eqref{eqn:test_functions} (see \cite[Section~1.6]{Fukushima2011Dirichlet}). Consequently, \ref{item:test_func} is the analogous condition in our setting. 

We will now show that Assumption~\ref{ass:nondegeneracy}\ref{item:test_func} implies that $\Pi_0$ is nonempty. The proof makes use of the theory of Dirichlet forms; see \cite{Fukushima2011Dirichlet} for an exposition of this theory. Let
\begin{equation} \label{eqn:mu_def}
d\mu(z) = p(z)\, dz
\end{equation}
and define the symmetric Dirichlet form $(\Ecal^0,D(\Ecal^0))$ as the closure on $L^2(F,\mu)$ of 
\begin{equation} \label{eqn:symmetric_dirichlet_form_def}
\Ecal^0 (u,v) := \int_F \nabla u^\top c \nabla v\, p; \qquad u,v \in C_c^\infty(F).
 \end{equation}
 By our regularity assumptions on $c$ and $p$ it follows from \cite[Theorem~1.12]{baur2013construction} that the corresponding semigroup $(T_t^0)_{t > 0}$ is strong Feller and that there exists a (a priori possibly explosive) solution $\P_z^0$ to the martingale problem corresponding to 
		\begin{equation} \label{eqn:generator_reversible}
			L^0 := \frac{1}{2}\div (c\nabla) + \frac{1}{2}\nabla \log p^\top c \nabla =  \frac{1}{2}\Tr(c\nabla^2) + \ell^\top c \nabla
		\end{equation}
for every starting point $z \in F$. From the form of the generator $L^0$ we see that $\P^0_z$ is the law of a weak solution to the SDE
\begin{equation} \label{eqn:tilde_p_dynamics}
	dZ_t = c(Z_t)\ell(Z_t)\, dt + c^{1/2}(Z_t)\, dW_t
\end{equation}
with initial condition $Z_0 = z$, where $W$ is a $(d+m)$-dimensional Brownian motion.

We will now establish ergodicity of the Dirichlet form and corresponding process. In particular, this excludes that the process explodes. First note that by \cite[Example~4.6.1]{Fukushima2011Dirichlet}, $(\Ecal^0, D(\Ecal^0))$ is irreducible, so to prove ergodicity we just have to establish recurrence. To this end define $\chi_n(x,y) = \varphi_n(x)\psi_n(y)$, where $\varphi_n,\psi_n$ are given in Assumption~\ref{ass:nondegeneracy}\ref{item:test_func}. Then $\chi_n \in C_c^\infty(F)$, $0 \leq \chi_n \leq 1$, $\lim_{n \to \infty} \chi_n = 1$ and 
\begin{equation} \label{eqn:recurrence}
\begin{split} \Ecal^0(\chi_n,\chi_n) &= \int_F \nabla \varphi_n(x)^\top c_X(x,y) \nabla \varphi_n(x) \psi_n^2(y) p(x,y)\, dx\, dy \\
& \quad + \int_F \nabla \psi_n(y)^\top c_Y(x,y) \nabla \psi_n(y) \varphi_n^2(x)p(x,y)\, dx\, dy \\
& \leq \int_E \nabla \varphi_n(x)^\top A(x) \nabla \varphi_n(x)\, dx + \int_D \nabla \psi_n(y)^\top B(y)\nabla \psi_n(y)\, dy.
\end{split} 
\end{equation} 
Hence, Assumption~\ref{ass:nondegeneracy}\ref{item:test_func} implies that $\lim_{n \to \infty} \Ecal^0(\chi_n,\chi_n) = 0$, from which we deduce that $\Ecal^0$ is recurrent, courtesy of \cite[Theorem~1.6.3]{Fukushima2011Dirichlet}. Ergodicity of the form and corresponding process now follows from \cite[Theorem 1.6.5(iii)]{Fukushima2011Dirichlet}.  From \eqref{eqn:tilde_p_dynamics} we see that $X$ has the correct volatility structure so it follows that $\P_z^0 \in \Pi_0$ for every $z \in F$.

\begin{remark} \label{rem:nonempty}
The above construction does not rely on the block diagonal form of $c = \mathrm{diag}(c_X,c_Y)$. Indeed, for a more general covariance matrix with possibly nonzero off-diagonal blocks, we can analogously define the Dirichlet form $\Ecal^0$ as in \eqref{eqn:symmetric_dirichlet_form_def} and obtain recurrence as in \eqref{eqn:recurrence} using the same test functions $\chi_n$. Indeed by positive-definiteness we have for any such matrix $c$ that 
\begin{align*} 
\frac12 \begin{pmatrix} \nabla \varphi_n(x) \psi_n(y) \\  \varphi_n(x) \nabla\psi_n(y) \end{pmatrix}^\top &   c(x,y) \begin{pmatrix} \nabla \varphi_n(x) \psi_n(y) \\  \varphi_n(x) \nabla\psi_n(y) \end{pmatrix} \\
& \leq \nabla \varphi_n(x)^\top c_X(x,y) \nabla \varphi_n(x) \psi_n^2(y) + \nabla \psi_n(y)^\top c_Y(x,y) \nabla \psi_n(y) \varphi_n^2(x)
\end{align*}
from which we obtain \eqref{eqn:recurrence} and then the admissibility of the measure corresponding to \eqref{eqn:tilde_p_dynamics}.
\end{remark}

\section{Results} \label{sec:results}
\subsection{Main result} \label{sec:main_result}
We are now ready to state our main results.  Recall $A$, $\lambda_{\Pi_K}$ and $\Dcal$ defined in \eqref{eqn:A_def},  \eqref{eqn:lambda_def} and \eqref{eqn:Dcal} respectively.
\begin{thm}[Main result] \label{thm:main}
	Let Assumptions~\ref{ass:regularity} and \ref{ass:nondegeneracy} be satisfied. Then there exists a unique (up to additive constant) $\hat \phi$ satisfying 
	\begin{equation} \label{eqn:argmin}  \hat \phi = \underset{\phi \in \Dcal}{\arg\min}\, \frac{1}{2}\int_E \frac{\Tr(A(x)\nabla^2 e^{\phi(x)})}{e^{\phi(x)}} \, dx.
	\end{equation} 
Define
	\begin{equation} \label{eqn:lambda_int}
	    \lambda = \frac{1}{2}\int_E \nabla \hat \phi^\top A\nabla \hat \phi 
	\end{equation} and the strategy
	\begin{equation} \label{eqn:hat_theta}
	    \hat \theta_t := \nabla \hat \phi(X_t); \quad t \geq 0.
	\end{equation}
Then for every compact set $K \subset E$ we have that $\lambda_{\Pi_K} = \lambda$. Moreover, $g(\hat \theta;\P) = \lambda$ for every $\P\in \Pi_K$, so that $\hat \theta$ is robust growth-optimal.
\end{thm} 
As discussed in Section~\ref{sec:heuristics}, proving Theorem~\ref{thm:main} includes two main parts: (i) establishing the existence of $\hat \phi$ and (ii) constructing a worst case measure $\widehat \P^K \in \Pi_K$. To establish the existence of $\hat \phi$ note that whenever $\phi \in \Dcal$ is compactly supported, integration by parts gives the identity
\[ \frac{1}{2}\int_E \frac{\Tr(A\nabla^2 e^{\phi})}{e^{\phi}} = \frac{1}{2}\int_E (\frac12A^{-1}\div A-\nabla \phi)^\top A( \frac12A^{-1} \div A - \nabla \phi) - \frac{1}{8} \int_E \div A^\top A^{-1}\div A. \]
The expression on the right hand side is more amenable to analysis than the one on the left hand side. Hence, as in \cite{kardaras2021ergodic}, we minimize the expression on the right hand side and show that the optimizer $\hat \phi$ satisfies \eqref{eqn:argmin} as well.

\begin{prop}[Existence of optimizer]
\label{prop:hat_phi}
	There exists a unique (up to additive constant) minimizer $\hat \phi$ to the variational problem \begin{equation} \label{eqn:unconstrained_min_problem}
		\inf_{\phi \in W^{1,2}_{\mathrm{loc}}(E)}\int_E ( \frac12A^{-1}\div A - \nabla \phi)^\top A( \frac12A^{-1} \div A - \nabla \phi).
	\end{equation}
	Moreover $\hat \phi$ belongs to $C^{2,\gamma'}(E)$ for some $0 < \gamma' \leq \gamma$ and satisfies the Euler--Lagrange equation
	\begin{equation} \label{eqn:phi_pde}
		\div(A(x)\nabla \hat \phi(x) - \frac12\div A(x)) = 0; \qquad x \in E.
	\end{equation}
 Additionally, $\hat \phi$ belongs to $\Dcal$ and satisfies \eqref{eqn:argmin}.
\end{prop}

\begin{remark}
Notice that the result also implies stability: assume that all quantities depend on a parameter $ \epsilon \in [0,1]$ and all assumptions holds for each $\epsilon$ in a continuous way (i.e.~all stated quantities and respective derivatives depend continuously on $\epsilon$). Then also $\hat \phi$ depends continuously on $\epsilon$ by elliptic regularity as well as $\lambda$ and $\hat \theta$ taking the results of Theorem \ref{T_pde_no_param} verbatim to the parameter dependent case.
\end{remark}

For the construction of the worst case measure $\widehat \P^K$, the dynamics of $X$ are pinned down by \eqref{eqn:X_dynamics}, but the dynamics of $Y$ need to be carefully selected. We take $c_Y$ from Assumption~\ref{ass:nondegeneracy} to be the instantaneous covariance matrix of $Y$ under $\widehat \P^K$, as is the case under $\P^0$. To select the drift of $Y$ we identify a function $\hat v$, which will be used to specify the dynamics. To state the next lemma we introduce some notation. Whenever $\tilde c_X \in \Ccal_K$ is given we write $\tilde \ell_X$ for $\frac{1}{2}\tilde c_X^{-1}\div_x \tilde c_X + \frac{1}{2} \nabla_x \log p$.
\begin{lem}[Existence of $\tilde c_X$ and $\hat v$]
\label{lem:hat_v}
For every compact $K \subset E$ there exists a $\tilde c_X \in \Ccal_K$ and a $\hat v: F \to \R$ satisfying the following properties:
\begin{enumerate}
    \item \label{item:hat_v_regularity}
    For a.e.\ $x \in E$, $\hat v(x,\cdot) \in W^{1,2}_{\mathrm{loc}}(D)$ and $\nabla_y \hat v \in L^q_{\mathrm{loc}}(F)$ for every $q \in [2,\infty)$. 
    \item \label{item:hat_v_pde} $\hat v$ is a weak solution to the PDE
   \[
        \div_y(c_Y(\ell_Y - \nabla_y \hat v)p) = -\div_x(\tilde c_X(\tilde \ell_X - \nabla\hat \phi)p) \quad \text{in } F,
\]
that is, 
    \begin{equation} \label{eqn:v_weak_pde}
    \int_F \div_x(\tilde c_X(\tilde \ell_X - \nabla \hat \phi)p)\, \psi - (\ell_Y - \nabla_y \hat v)c_Y \nabla_y \psi\, p = 0 \quad \text{for every }\psi \in C_c^1(F).
    \end{equation}
    \item \label{item:hat_v_bound} We have the inequality 
\begin{equation} \label{eqn:y_l2_bound}
\int_F \nabla_y \hat v^\top c_Y \nabla_y \hat v\, p \leq C \left(\int_F \frac{(\div_x(\tilde c_X (\tilde \ell_X - \nabla \hat \phi)p))^2}{\lambda_{\min}(c_Y)p} + \int_F \ell_Y^\top c_Y \ell_Y\, p \right) < \infty
\end{equation}
for a constant $C$ which only depends on the diameter of $D$.
\end{enumerate}
\end{lem}
Now we have all the ingredients to define a worst case measure.
\begin{thm}[Worst-case measure] \label{thm:worst_case_measure_existence}
	For every compact set $K \subset E$ and every $(x,y) \in F$ there exists a measure $\widehat \P^K_{(x,y)}$ on $(\Omega,\F)$ which is the law of a weak solution to the stochastic differential equation 
	\begin{equation} \label{eqn:worst_case_dynamics}
		\begin{split}
			&dX_t  = \tilde c_X(X_t,Y_t)\nabla\hat \phi(X_t)\, dt + \tilde c_X^{1/2}(X_t,Y_t)\, dW^X_t \\
			&dY_t  = c_Y(X_t,Y_t)\nabla_y\hat v(X_t,Y_t)\, dt + c_Y^{1/2}(X_t,Y_t)\, dW^Y_t
		\end{split}
	\end{equation}
 and satisfies $\widehat \P^K_{(x,y)}(X_0 = x, Y_0 =y) = 1$.
Here $W := (W^X,W^Y)$ is a standard $(d+m)$-dimensional Brownian motion, $c_Y$ is from Assumption~\ref{ass:nondegeneracy}, $\hat \phi$ is the optimizer from Proposition~\ref{prop:hat_phi} and $\tilde c_X$, $\hat v$ are given by Lemma~\ref{lem:hat_v}.

We additionally have that $\mu$, given by \eqref{eqn:mu_def}, is an invariant measure for $(X,Y)$ and for every locally bounded $h \in L^1(F,\mu)$,
\begin{equation} \label{eqn:ergodic_property}
	\lim_{T \to \infty} \frac{1}{T}\int_0^T h(X_t,Y_t)\, dt = \int_F hp; \quad \widehat \P^K_{(x,y)}\text{-a.s.}
\end{equation}
Thus the laws of $X_t$, $t \geq 0$, are tight under $\widehat \P^K_{(x,y)}$ and we have $\widehat \P^K_{(x,y)} \in \Pi_K$ for every $(x,y) \in F$.
\end{thm}

The construction of $\widehat \P^K_{(x,y)}$ and, in particular, verifying that it satisfies the ergodic property \eqref{eqn:ergodic_property} is a delicate matter. Our proof combines PDE techniques and recent results in the theory of generalized Dirichlet forms. The technical reason for introducing $K$-modifications and the class $\Pi_K$ is that, in general, it is not clear when the process \eqref{eqn:worst_case_dynamics} with $\tilde c_X$ replaced by $c_X$ is ergodic. However, if one can verify that 
\begin{equation} \label{eqn:ideal_bound}
\int_F \frac{(\div_x(c_X(\ell_X - \nabla \hat \phi)p))^2}{\lambda_{\min}(c_Y)p} < \infty
\end{equation}
then \eqref{eqn:worst_case_dynamics} with $\tilde c_X$ replaced by $c_X$ is ergodic and $K$-modifications are not needed. In a special case we can ensure that \eqref{eqn:ideal_bound} holds, yielding a refined version of Theorem~\ref{thm:main}. We state this case as a corollary.
\begin{cor}[$\Pi_0$ result]\label{cor:original}
Let Assumptions~\ref{ass:regularity} and \ref{ass:nondegeneracy} hold. Assume additionally that $A^{-1} \div A$ is the gradient of a function and that
\begin{equation} \label{eqn:grad_int_assumption}
    \int_F \frac{(\div_x(c_X(\ell_X - \frac12 A^{-1}\div A)p))^2}{\lambda_{\min}(c_Y)p} < \infty.
\end{equation}
Then $\lambda_{\Pi_0} = \lambda$, where $\lambda$ is given by \eqref{eqn:lambda_int}. Moreover, $g(\hat \theta;\P) = \lambda$ for every $\P \in \Pi_{0}$ so that $\hat \theta$ is robust growth-optimal.
\end{cor}
When $m=1$ and $A^{-1}\div A$ is the gradient of a function then any solution $\hat v$ to \eqref{eqn:v_weak_pde} satisfies
\begin{equation} \label{eqn:m=1}
\int_F (\ell_Y - \nabla_y \hat v)^\top c_Y(\ell_Y - \nabla_y \hat v)\,p = \int_F \frac{(\div_x(c_X(\ell_X - \frac12 A^{-1}\div A)p))^2}{\lambda_{\min}(c_Y)p}.
\end{equation} 
Finiteness of the left hand side of \eqref{eqn:m=1} is crucially used to show recurrence of the worst-case measure. Hence, using our methods and at this level of generality, the integrability condition \eqref{eqn:grad_int_assumption} cannot be improved upon when studying the robust problem without $K$-modifications.

The proofs of Proposition~\ref{prop:hat_phi}, Lemma~\ref{lem:hat_v}, Theorem~\ref{thm:worst_case_measure_existence} and Corollary~\ref{cor:original} are contained in Appendix~\ref{sec:proofs}, but we give a proof of Theorem~\ref{thm:main} here.
\begin{proof}[Proof of Theorem~\ref{thm:main}]
Since $\hat \phi \in \Dcal$, in view of the discussion of Section~\ref{sec:heuristics}, we have that the asymptotic growth rate of $\hat \theta$ is the same (finite) value for every admissible measure. This yields the lower bound $\lambda_{\Pi_K} \geq g(\hat \theta;\widehat \P^K)$, where we chose $\widehat \P^K$ constructed in Theorem~\ref{thm:worst_case_measure_existence} for concreteness and omitted the starting point from the notation for simplicity. Conversely, since $\hat \theta$ is growth-optimal under $\widehat \P^K \in \Pi_K$ we obtain the upper bound $\lambda_{\Pi_K} \leq g(\hat \theta;\widehat \P^K)$. This establishes the robust growth-optimality of $\hat \theta$ and so it just remains to prove the formula \eqref{eqn:lambda_int}.

To this end note that under $\widehat \P^K$ we have that 
\begin{equation} \label{eqn:worst_case_wealth}
\log V^{\hat \theta}_T = M_T + \frac{1}{2}\langle M \rangle_T,
\end{equation}
where $M_T = \int_0^T \nabla \hat \phi(X_t)^\top \tilde c^{1/2}_X(X_t,Y_t)\, dW_t^X$.
By the ergodic property we have $\widehat \P^K$-a.s.\ that 
\begin{equation} \label{eqn:qv_limit} \lim_{T \to \infty} \frac{\langle M \rangle_T}{T} = \int_F \nabla \hat \phi(x)^\top \tilde c_X(x,y)\nabla \hat \phi(x)p(x,y)\, dx\, dy = \int_E \nabla \hat \phi(x)^\top A(x) \nabla \hat \phi(x)\, dx < \infty.
\end{equation} 
Using \cite[Lemma~1.3.2]{fernholz2002stochastic} we obtain that $M_T/T \to 0$, $\widehat \P^K$-a.s.\ as $T \to \infty$. Hence we see from \eqref{eqn:worst_case_wealth} and \eqref{eqn:qv_limit} that $g(\hat \theta;\widehat \P^K) = \frac12 \int_E \nabla \hat \phi^\top A \nabla \hat \phi$ which completes the proof.
\end{proof}

\subsection{The case of fully specified joint covariation structure} \label{sec:main_extended}

 The problem we considered thus far fixes only the covariation structure $c_X$ of $X$ and the joint invariant density $p$ of $(X,Y)$. The local dynamics of $Y$ and its local interaction with $X$ is otherwise unrestricted. This leads to the very large classes of measures $\Pi_K$ of Definition~\ref{def:upsilon-2} that we are robust over. In particular, the measure $\widehat \P^K$ (where we omit the initial value for convenience) constructed in Theorem~\ref{thm:worst_case_measure_existence}, which has the property that $\langle X, Y\rangle = 0$, belongs to $\Pi_K$ and is able to serve as a worst-case measure. However, depending on the choice of factor process $Y$, more information about its dynamics might be accessible for estimation. If empirical measurements imply that $Y$ is a semimartingale with $\langle X, Y \rangle \ne 0$,\footnote{This is something we can expect for certain choices of factor process. As one example consider the case where $Y$ is taken to be the level of a market volatility measuring index such as the VIX. 
 } then the classes $\Pi_K$ of Definition~\ref{def:upsilon-2} would seem too large. In particular, the measure $\widehat \P^K$ ought not to be admissible.
 
 In this section we study the robust problem where $Y$ is assumed to be a semimartingale and the entire instantaneous joint covariation matrix $c$ of $X$ and $Y$ is specified as an input along with the joint invariant density $p$. As such $\langle X, Y \rangle$ and $\langle Y \rangle$ can additionally be specified as inputs and, in general, $\widehat \P^K$ of Theorem~\ref{thm:worst_case_measure_existence} would no longer be admissible. We thank two anonymous referees for suggesting this extension.
 
 Our main finding in this setting is that, remarkably, under the appropriate minor modifications to Assumption~\ref{ass:nondegeneracy} given by Assumption~\ref{ass:nondegeneracy_joint} below, the strategy $\hat \theta$ of Theorem~\ref{thm:main} remains optimal and the corresponding robust optimal growth rate is still given by $\lambda$. In particular, even though the \emph{entire} covariation structure is specified, the covariation of $Y$ and its joint covariation with $X$ do not impact the robust optimal strategy or growth rate.   

 Concretely, in addition to $p$, we take $c:F \to \mathbb{S}^{d+m}_{++}$ as an input and $c(X_t,Y_t)$ will be the joint instantaneous covariation matrix of $X$ and $Y$. We canonically write $c$ in block form as
\begin{equation} \label{eqn:c_block} 
c(z) =  \begin{bmatrix} c_X(z) & c_{XY}(z) \\ c_{XY}(z)^\top & c_Y(z)\end{bmatrix}, \quad z \in F,
\end{equation}
where $c_{XY}$(z) is a $(d \times m)$-dimensional matrix. Given that $c_X$ and $p$ satisfy Assumption~\ref{ass:regularity}, we continue to define the averaged instantaneous covariance function $A$ by \eqref{eqn:A_def}. We then consider a slightly different notion of $K$-modification, where now the joint covariation matrix $c$ is modified outside a compact set $K \subset F$ to yield another joint covariation matrix $\tilde c : F \to \S^d_{++}$. More precisely, we define
\[
\Ccal_K^F = \left\{ \tilde c \in W^{1,\infty}_{\mathrm{loc}}(F): \tilde c = c \text{ on } K \text{ and } \int_D \tilde c_X(x,y)p(x,y)\, dy = A(x) \text{ for } x \in E \right\},
\]
where $\tilde c_X$ refers to the upper left $d\times d$ block of $\tilde c$ in accordance with \eqref{eqn:c_block}. We refer to an element $\tilde c_X$ of $\Ccal_K^F$ as a $K$-modification of $c$. 
The associated class of measures $\Pi_K^c$ is then given as in Definition~\ref{def:upsilon-2} with a modified first condition
\begin{itemize} \label{item:modified_i}
  \item[(i'')] $Z = (X,Y)$ is a $\P$-semimartingale with covariation process $\langle Z \rangle = \int_0^\cdot \tilde c(Z_t)\, dt$ for some $\tilde c \in \Ccal_K^F$.
\end{itemize}
The robust optimal growth rate in this setting is
\[\lambda_{\Pi^c_K} = \sup_{\theta \in \Theta} \inf_{\P \in \Pi_K^c} g(\theta;\P).\]
To state the analogue of Theorem~\ref{thm:main} in this setting we first need an analogue of Assumption~\ref{ass:nondegeneracy}.
\begin{assum}[Nondegeneracy assumptions v2.] 
\label{ass:nondegeneracy_joint}
    As before denote by $\ell = \frac{1}{2} c^{-1}\mathrm{div}c + \frac{1}{2}\nabla \log p$ and decompose $\ell = (\ell_X,\ell_Y)$. We also define the block diagonal matrix $c^0 = \mathrm{diag}(c_X,c_Y)$. Set $\xi = (c^0)^{-1}c\ell$ and similarly write $\xi = (\xi_X,\xi_Y)$. We assume that
	\begin{enumerate}
		\item  $c_Y\in W^{2,\infty}_{\mathrm{loc}}(F)$ and  $c_Y(x,\cdot)p(x,\cdot)$ is locally Lipschitz continuous for every $x \in E$.
		\item  For every fixed $x \in E$, there exists a constant $k_x \geq 0$ and a concave function $\rho_x:D \to (0,\infty)$ such that $\lambda_{\mathrm{min}}(c_Y(x,y))p(x,y) = \rho_x(y)^{k_x}$. \label{item:poincare}
		\item    For every fixed $x \in E$ and every $C \in \R$, $b \in \R^d$, $M \in \R^{d \times d}$ we have 
\[		    \int_D\frac{ \div_x(c_X\xi_X)^2  + (\xi_X^\top c_X (\nabla_x \log p + b))^2  + \Tr(c_XM)^2 + C}{\lambda_{\mathrm{min}}(c_Y)}p < \infty.
\]
\item $\int_{F} \xi^\top c^0 \xi p < \infty$.
		\item  $\int_F \div_x(c_X\ell_Xp)\, < \infty$.
  \item There exist functions $\varphi_n \in C_c^\infty(E)$ and $\psi_n \in C_c^\infty(D)$ satisfying $0 \leq \varphi_n, \psi_n \leq 1$, $\lim_{n \to \infty} \varphi_n = \lim_{n \to \infty} \psi_n = 1$ and 
\begin{equation} \label{eqn:test_functions}
    \lim_{n \to \infty} \int_E \nabla \varphi_n^\top A \nabla \varphi_n(x)\, dx = \lim_{n \to \infty} \int_D \nabla \psi_n^\top B\, \nabla \psi_n(y)\, dy = 0,
\end{equation}
where $A$ is given by \eqref{eqn:A_def} and $B(y) = \int_E c_Y(x,y)p(x,y)\, dx$ is assumed to be finite for almost every $y \in D$.
    \end{enumerate}
    \end{assum}
\begin{remark}
As before, Assumption~\ref{ass:nondegeneracy_joint}\ref{item:poincare} is needed to establish the weighted Poincar\'e inequality 
\begin{equation} \label{eqn:poincare}
\begin{split}& \int_U \left(u(y)-\frac{\int_U u(y)\lambda_{\mathrm{min}}(c_Y(x,y))p(x,y)dy}{\int_U \lambda_{\mathrm{min}}(c_Y(x,y))p(x,y)dy}\right)^2\lambda_{\mathrm{min}}(c_Y(x,y))p(x,y)\, dy 
\\ & \hspace{5cm}\leq \frac{\mathrm{Diam}(U)^2}{\pi^2}\int_U \|\nabla u(y)\|^2 \lambda_{\mathrm{min}}(c_Y(x,y))p(x,y)dy
\end{split}
\end{equation}
for every $ x \in E$, any convex domain $U \subset D$ and any Lipschitz function $u$ (see \eqref{eqn:weighted_poincare_lipschitz} and Theorem~\ref{T_pde_no_param} in Appendix A). As such, Assumption~\ref{ass:nondegeneracy_joint}\ref{item:poincare} (and similarly Assumption~\ref{ass:nondegeneracy}\ref{item:assum_new_4}) can be replaced by any other condition on $\lambda_{\mathrm{min}}(c_Y(x,y))p(x,y)$ that ensures the weighted Poincar\'e inequality \eqref{eqn:poincare} holds. In the PDE literature, weighted Poincar\'e inequalities have also been shown when the weight function belongs to an $A_p$ Muckenhoupt class. However, in Example~\ref{ex:beta} this
requirement resulted in stricter conditions on the parameters than the requirements of
Assumption~\ref{ass:nondegeneracy_joint}\ref{item:poincare}.
\end{remark}

The difference between Assumption~\ref{ass:nondegeneracy_joint} and Assumption~\ref{ass:nondegeneracy} lies in items \ref{item:assum_new_3} and \ref{item:technical_1}, where $\xi$ and $c^0$ replace $\ell$ and $c$. Additionally note that if $c_{XY} = 0$ then $c^0 = c$ and $\xi = \ell$. Hence Assumption~\ref{ass:nondegeneracy_joint} reduces to Assumption~\ref{ass:nondegeneracy} in this case.
We are now ready to state our main result in this setup.
\begin{thm}[Main theorem for specified joint covariation] \label{thm:main_extended}  
	Let Assumption~\ref{ass:regularity} and Assumption~\ref{ass:nondegeneracy_joint} be satisfied. Let $\hat \phi, \lambda$ and $\hat \theta$ be as in \eqref{eqn:argmin}, \eqref{eqn:lambda_int} and \eqref{eqn:hat_theta} respectively. 
Then for every compact set $K \subset F$ we have $\lambda_{\Pi_K^c} = \lambda$. Moreover, $g(\hat \theta;\P) = \lambda$ for every $\P\in \Pi_K^c$, so that $\hat \theta$ is robust growth-optimal.
\end{thm} 
The proof of Theorem~\ref{thm:main_extended} follows in a similar fashion as Theorem~\ref{thm:main}. Indeed, since $\Pi_K^c \subset \Pi_K$ it follows that $\lambda_{\Pi_K^c} \geq \lambda_{\Pi_K} = \lambda$ and $g(\hat \theta;\P) = \lambda$ for every $\P \in \Pi_K^c$. Hence, to establish the reverse inequality we again aim to construct a measure $\widehat \P^K_c \in \Pi_K^c$ such that $\hat \theta$ is growth-optimal over all strategies under the measure $\Pi_K^c$. Note that the measure $\widehat \P^K$ of Theorem~\ref{thm:worst_case_measure_existence} is no longer admissible due to the requirement that the joint covariance structure of $X$ and $Y$ needs to be given by (a $K$-modification of) $c$. Nevertheless, for an explicit $K$-modification $\tilde c$ of $c$, given by equation \eqref{eqn:K-mod-extended} in the appendix, and for a function $\hat v_c$ possessing certain properties discussed below, we are able to construct an admissible measure $\widehat \P^K_c$ under which $(X,Y)$ has dynamics
\begin{equation} \label{eqn:worst_case_extended}
\begin{split}
     d\begin{pmatrix} X_t \\ Y_t \end{pmatrix} & = \begin{pmatrix} \tilde c_X(X_t,Y_t)\nabla \hat \phi(X_t) \\ \tilde c_Y(X_t,Y_t) \nabla_y \hat v_c(X_t,Y_t)\end{pmatrix}dt + \tilde c^{1/2}(X_t,Y_t)\,dW_t \\
     & = \tilde c^0(X_t,Y_t) \begin{pmatrix} \nabla \hat \phi(X_t) \\ \nabla_y \hat v_c(X_t,Y_t)\end{pmatrix}dt + \tilde c^{1/2}(X_t,Y_t)\,dW_t,
\end{split}
\end{equation}
where $\tilde c^0 = \mathrm{diag}(\tilde c_X,\tilde c_Y).$
The instantaneous covariation matrix is clearly $\tilde c$ and the dynamics of $X$ in \eqref{eqn:worst_case_extended} ensure that $\hat \theta$ is robust growth optimal. 
\begin{remark}
    Although the choice of matrix square root $\tilde c^{1/2}$ does not impact the growth-optimal strategy, it is convenient to take it to be a block lower triangular square root of $\tilde c$, namely
    \[\tilde c^{1/2} = \begin{bmatrix} 
     \tilde c_X^{1/2} & 0 \\
     \sigma_{XY} & \sigma_Y 
    \end{bmatrix} \]
    where $\tilde c_X^{1/2}$ is a matrix square root of $\tilde c_X$ and $\sigma_{XY}$, $\sigma_Y$ are chosen so that $\tilde c^{1/2} (\tilde c^{1/2})^\top = \tilde c$. Then \eqref{eqn:worst_case_extended} becomes
    \begin{align*}
    dX_t&  = \tilde c_X(X_t,Y_t)\nabla \hat \phi(X_t)\, dt + \tilde c_X^{1/2}(X_t,Y_t)\, dW_t^X \\
    dY_t & = \tilde c_Y(X_t,Y_t) \nabla_y \hat v_c(X_t,Y_t)\,dt + \sigma_{XY}(X_t,Y_t)\, dW_t^X + \sigma_Y(X_t,Y_t)\, dW_t^Y, 
    \end{align*}
    where we decomposed $W = (W^X,W^Y)$. This representation directly relates the dynamics of $X$ to \eqref{eqn:X_dynamics} ensuring growth-optimality of $\hat \theta_t = \nabla \hat \phi(X_t)$.
\end{remark}

To establish admissibility of \eqref{eqn:worst_case_extended} it just remains to choose $\hat v_c$ in such a way that $(X,Y)$ is ergodic with invariant density $p$. As in the proof of Theorem~\ref{thm:main}, the construction of $\widehat \P_K^c$ hinges on finding $\hat v_c$ satisfying certain properties. That such a $\hat v_c$ exists is guaranteed by the following lemma, which is an analogue of Lemma~\ref{lem:hat_v}.

\begin{lem}[Existence of $\hat v_c$] \label{lem:hat_vc} Let $\tilde \ell = \frac{1}{2}\tilde c^{-1} \mathrm{div} \tilde c + \frac{1}{2}\nabla \log p$ and let $\tilde \xi = (\tilde c^0)^{-1}\tilde c \tilde \ell$.
There exists a $\hat v_c:F \to \R$ satisfying the following properties.
\begin{enumerate}[noitemsep]
 \item  \label{item:hat_vc1}  
    For every $x \in E$, $\hat v_c(x,\cdot) \in W^{1,2}_{\mathrm{loc}}(D)$ and $\nabla_y \hat v_c \in L^q_{\mathrm{loc}}(F)$ for every $q \in [2,\infty)$. 
    \item $\hat v_c$ is a weak solution to the PDE \label{item:hat_vc2}
    \begin{equation}  \label{eqn:hat_vc_weak}
        \div_y(\tilde c_Y(\tilde \xi_Y - \nabla_y \hat v_c)p) = -\div_x(\tilde c_X(\tilde \xi_X - \nabla\hat \phi)p) \quad \text{in } F.
    \end{equation}
    \item We have the inequality  \label{item:hat_vc3}
\begin{equation*} 
\int_F \nabla_y \hat v_c^\top c_Y \nabla_y \hat v_c\, p \leq C \left(\int_F \frac{(\div_x(\tilde c_X (\tilde \xi_X - \nabla \hat \phi)p))^2}{\lambda_{\min}(c_Y)p} + \int_F \tilde \xi_Y^\top \tilde c_Y \tilde \xi_Y\, p \right) < \infty
\end{equation*}
for a constant $C$ which only depends on the diameter of $D$.
\end{enumerate}
\end{lem}

The existence of $\hat v_c$ satisfying Lemma~\ref{lem:hat_vc} \ref{item:hat_vc1}-\ref{item:hat_vc3} allows us, in the same way as in the proof of Theorem~\ref{thm:worst_case_measure_existence}, to ensure that $\widehat \P^K_c$ given as the law of the diffusion \eqref{eqn:worst_case_extended} is indeed ergodic with invariant density $p$. As such, it is a member of $\Pi_K^c$ and able to serve as a worst-case measure. This establishes the reverse inequality $\lambda_{\Pi_K^c} \leq \lambda$ and yields Theorem~\ref{thm:main_extended}. The details of the proof are in Appendix~\ref{sec:proofs_extended}.
    

\section{Discussion} \label{sec:discussion}
In this section we discuss the results of Section~\ref{sec:results} and their financial interpretation.

\subsection{Dependence of $\hat \theta$ and $\lambda$ on $Y$} \label{sec:discussion_Y}
Note that the optimal strategy $\hat \theta$ from Theorem~\ref{thm:main} is functionally generated, and hence in feedback form, but only depends on $X$, not $Y$. Moreover, by Theorem~\ref{thm:main_extended}, this remains true even if the joint covariation structure of $X$ and $Y$ is fixed as an input to the problem. Remarkably, this is the case even though $Y$ may be (partially) observable. Indeed, if $y \mapsto c_X(x,y)$ is invertible on $D$ for every $x 
\in E$, then one can back out $Y_t$ from observing $X_t$ and $c_X(X_t,Y_t)$. However, what Theorems~\ref{thm:main} and \ref{thm:main_extended} show is that knowing the trajectory of $Y$ does not improve performance in an adversarially chosen measure. Indeed, the measure $\widehat \P^K$ constructed in Theorem~\ref{thm:worst_case_measure_existence} (or $\widehat \P^K_c$ in the context of Section~\ref{sec:main_extended}) is precisely such a measure. Therefore any strategy, including those that depend on the trajectory of $Y$, cannot perform better than $\hat \theta$ under this measure. 
This is surprising since $Y$ may be coupled with $X$ in a nontrivial way, both locally through the joint covariation matrix $c$ and in the long-term via the joint invariant density $p$. 

The reasons  for this are threefold concerning the special properties of the growth-rate criterion, the non-investability of $Y$ and the absence of a local restriction on the drift of $Y$. Indeed, for fixed dynamics of $X$ given by a specified measure $\P$ the growth-optimal strategy is \emph{entirely} determined by the covariation matrix and drift vector of the asset process, irrespective of what other factors it depends on. As such under any measure $\P$ where $X$ has dynamics \eqref{eqn:X_dynamics}, the growth-optimal strategy is given by $\hat \theta$. This is true regardless of the dynamics of $Y$ since $Y$ is not investable.

However, to be consistent with empirical observations and to satisfy the admissibility criteria of our class of measures, we need to find such a measure where additionally $X$ and $Y$ have joint invariant density $p$. Even when the covariation structure of $Y$ is restricted, as is the case in Section~\ref{sec:main_extended}, the local drift dynamics are not. Hence with one condition to be met (the invariant density needs to be $p$) and one degree of freedom (the drift of $Y$) it then becomes possible -- though it is a difficult analytical problem -- to construct an admissible worst-case measure. If the drift of $Y$ was also restricted then one may expect the conclusions to change and for the robust-optimal strategy to depend on $Y_t$. We leave this interesting question for future research.  

Lastly, we discuss what influence $Y$ has on the problem. Clearly, information about the distribution of $Y$ does enter into and influences the optimal strategy $\hat \theta$ and robust optimal growth-rate $\lambda$. Indeed, the function $\hat \phi$, which specifies $\hat \theta$, is determined by $A$ through \eqref{eqn:phi_pde}. $A$ itself, given by \eqref{eqn:A_def}, is the average of the instantaneous covariance coefficient $c_X(x,\cdot)$ with respect to the density $p(x,\cdot)$. In the context of Section~\ref{sec:main_extended} we stress that in no way do the inputs $c_{XY}$ and $c_Y$ enter into the strategy; not even through their averages.


\subsection{Relationship to \cite{kardaras2021ergodic}}
Since the optimal strategy $\hat \theta$ only depends on $X$, it is also the optimal strategy for the problem considered in \cite{kardaras2021ergodic} when one takes appropriate inputs $\bar c_X(x)$ and $\bar p(x)$ satisfying the required assumptions in \cite{kardaras2021ergodic}. Indeed, up to technical conditions, the requirement on $\bar c$ and $\bar p$ that is needed to ensure that $\hat \phi$ of Proposition~\ref{prop:hat_phi} determines the optimal strategy in the problem considered in \cite{kardaras2021ergodic} is to have
\begin{equation} \label{eqn:KR_A}\bar c_X(x)\bar p(x) = A(x); \qquad \text{for all } x \in E.
\end{equation}
One particular instantiation of this, which has a clear interpretation, is to set $\bar p(x) = \int_D p(x,y)\, dy$ and $\bar c_X(x) = A(x)/\bar p(x)$. Then $\bar p$ is the marginal density of $X$ in our setting and using the definition of $A$, we have 
\[ \bar c_X(x) = \int_D c_X(x,y) \frac{p(x,y)}{\bar p(x)}\, dy = \E[c_X(X,Y)|X=x],\]
where $(X,Y)$ have joint density $p$. The matrix $\bar c_X(x)$ can then be viewed as an \emph{effective volatility matrix} for $X$ as $Y$ was averaged out, conditional on $X=x$, with respect to the invariant measure.

We now remark that Theorem~\ref{thm:main} can be viewed as a generalization of the main theorem in \cite{kardaras2021ergodic}. Indeed, robustness of the optimal strategy is shown over a much larger class of measures by allowing the volatility to depend on a factor process. Indeed, suppose that, instead of the setup of our paper, we take the setup of \cite{kardaras2021ergodic} with inputs $\bar c_X(x)$ and $\bar p(x)$. The admissible class of measures $\Pi$ of their paper is the class of measures satisfying items \ref{item:KR1}-\ref{item:KR3} stated in the introduction.
Then our Theorem~\ref{thm:main} shows that (up to technical conditions) the strategy $\hat \theta$ is not only optimal for the class of measures $\Pi$ but also for the classes $\Pi_K$ of Definition~\ref{def:upsilon} depending on \emph{any} factor process $Y$ taking values in \emph{any} bounded convex open set $D \subset \R^m$ for \emph{any} $m$ with \emph{any} inputs $c_X$ and $p$ as in Assumption~\ref{ass:regularity} as long as the average quantity $A$ given by \eqref{eqn:A_def} coincides with the one given by \eqref{eqn:KR_A}. 

We also point out that ellipticity of $c_X$ can possibly be relaxed towards hypo-ellipticity, which still allows to establish ergodicity of the involved process. The analysis of this aspect is left for future research. We also note that for several parts of the analysis only ellipticity properties of $A$ matter, which in presence of stochastic factors can be met (due to averaging) even if $c_X$ is not invertible everywhere.

This generalized result provides a more complete answer to, and connects, the two questions posed by Fernholz \cite[Problems~3.1.7 \& 3.1.8]{fernholz2002stochastic}. Indeed, our result shows that assuming only the presence of sufficient volatility and stability in the market, positive long-term relative growth is achievable and, additionally, that the optimal strategy achieving this growth rate is functionally generated. Although our result takes explicitly given inputs $c_X$ and $p$, Example~\ref{eg:exogenous} below shows that this result holds in substantial generality even when an explicit form for the volatility is not assumed. Example~\ref{eg:exogenous} is outside the scope of \cite{kardaras2021ergodic} but can be handled by our setup.

\subsection{$K$-modifications} \label{sec:K-mod}
Theorem~\ref{thm:main} requires $K$-modifications. That is, we allow modifications to the input matrix $c_X$ outside of a sufficiently large compact set $K \subset E$. One reason for for introducing these modifications is that the classes $\Pi_K$ can be seen as interpolating between the smallest class $\Pi_0$, which entirely fixes the volatility structure of $X$, and the much larger class $\Pi_{\emptyset}$, which up to the averaging condition that characterizes $\Ccal_\emptyset$ imposes few restrictions on the form of $\langle X \rangle$. Remarkably, across this wide range of classes, the optimal strategy $\hat \theta$ and the robust optimal growth rate $\lambda$ remain independent of $K$. This finding uncovers another layer of robustness -- across classes of measures -- for this problem.

Another reason we introduce $K$-modifications is technical in nature. Indeed, it allows us to construct the admissible worst-case measure $\widehat \P_K$ of Theorem~\ref{thm:worst_case_measure_existence} without imposing strong implicit conditions as in \eqref{eqn:grad_int_assumption}. Nevertheless, Corollary~\ref{cor:original} supports the case that this is a technical matter since, under additional analytic assumptions, our main result continues to hold without $K$-modifications.  Moreover, the $K$-modification $\tilde c_X$ of $c_X$ appearing in Lemma~\ref{lem:hat_v} and Theorem~\ref{thm:worst_case_measure_existence} is explicitly given by \eqref{eqn:K-mod} in Appendix~\ref{sec:proofs}. 

Additionally we point out that the amount of time $X$ spends in the modified region can be made arbitrarily small under all measures $\P \in \Pi_{\emptyset}$ simultaneously. Indeed, let $\varepsilon > 0$ be given and choose a compact $K \subset E$ such that $\mu(K \times D) \geq 1-\varepsilon$, where we recall that $\mu$ given by \eqref{eqn:mu_def} is the probability measure on $F$ with density $p$. Then for any $\P \in \Pi_{\emptyset}$ the ergodic property yields
\[\lim_{T \to \infty} \frac{1}{T}\int_0^T 1_{\{X_t \in E\setminus K\}}\, dt = \mu((E\setminus K) \times D) \leq\varepsilon.\]
Consequently, when $K$ is large, estimating the instantaneous covariation matrix accurately for $x \in E \setminus K$ is infeasible as the process rarely enters this region. As such, allowing $K$-modifications does not restrict the empirical estimation of the inputs. One can view our result as having additional robustness with respect to data sensitivity since the robust-optimal strategy $\hat \theta$ is independent of $K$. An analogous discussion holds for Theorem~\ref{thm:main_extended} and the $K$-modifications used there.

\section{Examples} \label{sec:examples}
	We now consider a few examples. Where not explicitly verified we assume that the inputs $c_X$ and $p$ are such that Assumptions~\ref{ass:regularity} and \ref{ass:nondegeneracy} are satisfied.
	\subsection{The gradient case}
	Analogously to \cite{kardaras2021ergodic}, when $A^{-1} \div A(x) = \nabla h(x)$ for some function $h$, we see from \eqref{eqn:phi_pde} that $\hat \phi = h/2$. Consequently, this is an important special case as it yields a fairly explicit optimal strategy, whereas otherwise solving \eqref{eqn:phi_pde} can be a difficult numerical problem, especially when $d$ is large. All of the examples below are cases where $A^{-1}\div A$ is a gradient.
	\subsection{The one-dimensional case}
	When $d = 1$ we have that 
	$A^{-1}\div A(x) = (\log A)'(x)$ so that $\hat \phi = \frac12\log A$ . Expanding using the definition of $A$ we see that
	\[(\log A)'(x) = \frac{\int_D \partial_x(c_Xp)(x,y)\, dy}{\int_D c_Xp(x,y)\, dy} = \int_D \partial_x \log (c_Xp) \frac{c_Xp(x,y)}{\int_D c_Xp(x,w)\, dw}\, dy.\]
	Consequently we can write
	\[\hat \phi(x) = \frac12 \E[\partial_x \log (c_Xp(X,Y))|X=x],\]
	where the random variables $(X,Y)$ have joint density proportional to $c_Xp$. In the special case when $X$ and $Y$ are independent under this measure; i.e.\ if $c_Xp(x,y) = f_X(x)f_Y(y)$ for some functions $f_X$,$f_Y$ then this simplifies to $\hat \phi(x) = \frac12\E[\partial_x \log (c_Xp(x,Y))]$.
	\subsection{A tractable class of models}
Assume that $p(x,y) = p_X(x)p_Y(y)$ and set
\[c_X^{ij}(x,y) = \begin{cases}  f_{ii}(x_{-i},y)f_i(x^i)g(x)h_{ii}(y), & i= j, \\
f_{ij}(x_{-ij},y)f_i(x^i)f_j(x^j)g(x)h_{ij}(y), & i \ne j 
\end{cases} \] for some functions  $p_X,p_Y,f_{ij}, f_i, g,h_{ij}$. Here $x_{-i}$ is a $(d-1)$-dimensional vector obtained from $x$ by removing the $i^\text{th}$ component and $x_{-ij}$ is a $(d-2)$-dimensional vector obtained from $x$ by removing the $i^{\text{th}}$ and $j^{\text{th}}$ components. In this case we have that 
\[A^{ij}(x) = \begin{cases} \bar f_{ii}(x_{-i})f_i(x^i)g(x)p_X(x), & i = j, \\  \bar f_{ij}(x_{-ij})f_i(x^i)f_j(x^j)g(x)p_X(x), &  i \ne j \end{cases}, \] where $\bar f_i(x_{-i}) = \int_D f_{ii}(x_{-i},y)h_{ii}(y)p_Y(y)\, dy$ and $\bar f_{ij}(x_{-ij}) = \int_D f_{ij}(x_{-ij},y)h_{ij}(y)p_Y(y)\, dy$. This matrix is reminiscent of the one introduced in \cite{itkin2020robust}. A direct calculation shows that $A^{-1}\div A$ is a gradient and as a consequence we obtain that
\[\hat \phi(x) = \sum_{i=1}^d \frac{1}{2} \log f_i(x^i) + \frac{1}{2}\log g(x) + \frac{1}{2}\log p_X(x).\]

\subsection{Exogenous stochastic factor and ergodic independence} \label{eg:exogenous}
Here we look at the case when $c_X$ only depends on $y$ and $p(x,y) = p_X(x)p_Y(y)$. Then
\[A^{ij}(x) = \bar A^{ij}p_X(x),\]
where $\bar A^{ij} = \int_D c_X^{ij}(y)p_Y(y)$ is a constant matrix.
Consequently $A^{-1} \div A(x) = \nabla \log p_X(x)$ and so $\hat \phi(x) = \frac{1}{2}\log p_X(x)$.
The robust growth rate is given by
\[\lambda = \int_E \nabla \log p_X(x)^\top \bar A\,  \nabla \log p_X(x)p_X(x)\, dx.\]
Importantly, the optimal strategy does not depend on $c_X$ or $p_Y$. When we take the volatility as exogenously given and only impose some asymptotic structure on its behaviour, which is independent of the asymptotic behaviour of $X$, the best strategy in a robust setting only depends on the stability properties of $X$. One interpretation for this example is volatility uncertainty. The other quantities in this case do not yield additional information that can be exploited under the admissible adversarial measure $\widehat \P^K \in \Pi_K$. Moreover, aside from the case when $p_X$ is constant, we have $\lambda > 0$ so that strictly positive robust asymptotic growth is achievable even when $c_X$ does not depend on $x$. Although the optimal strategy does not depend on $c_X$ or $p_Y$, the robust growth rate does through the average value $\bar A$.

\subsection{One-dimensional Beta densities} \label{ex:beta}
Set $E = D = (0,1)$ and take
\begin{align*} c_X(x,y) & = \sigma^2 x^{b_1}(1-x)^{b_2}(x+y)^{q_1}(2-x-y)^{q_2}, \\
p(x,y) & = \frac{x^{a_1-1}(1-x)^{a_2-1}y^{\alpha_1 -1}(1-y)^{\alpha_2-1}}{B(a_1,a_2)B(\alpha_1,\alpha_2)},
\end{align*}
where $\sigma^2 > 0$ and the other parameters are such that
\begin{equation} \label{eqn:param_constraints}
	q_1, q_2 \geq 0, \quad a_1, a_2 > 0 \qquad b_1, b_2\geq 1, \qquad \alpha_1,\alpha_2 > 1, \qquad b_1 + a_1 > 2, \quad b_2 + a_2 >2.
\end{equation} 
Here $B(\cdot,\cdot)$ is the Beta function. In this case 
\[\ell_X(x,y) = \frac{1}{2}\left(\frac{b_1+a_1-1}{x} - \frac{b_2 + a_2-1}{1-x} + \frac{q_1}{x+y} - \frac{q_2}{2-x-y}\right).
\] 
The parameter constraints \eqref{eqn:param_constraints} ensure that $\int_F \ell_X^2 c_X\, p < \infty$ and $\int_F \partial_x(c_X\ell_Xp) < \infty$.
Next set $c_Y(x,y) = y^{\beta_1}(1-y)^{\beta_2}$, where 
\begin{equation} \label{eqn:beta_bounds}
    2 - \alpha_i < \beta_i < 2\alpha_i -1; \qquad \text{ for } i=1,2.
\end{equation} Note that since $\alpha_i > 1$, the interval $(2-\alpha_i, 2\alpha_i-1)$ is nonempty. Next, we have that 
\[c_Y(x,y)p(x,y) = \frac{x^{a_1-1}(1-x)^{a_2-1}y^{\alpha_1 + \beta_1 -1}(1-y)^{\alpha_2 + \beta_2-1}}{B(a_1,a_2)B(\alpha_1,\alpha_2)}.\]
Since the functions $y \mapsto y^\gamma(1-y)$ and $y\mapsto y(1-y)^\gamma$ are concave for $\gamma \in [0,1]$ it follows that Assumption~\ref{ass:nondegeneracy}\ref{item:assum_new_4} is satisfied. Moreover, for every fixed $x \in (0,1)$, $\int_0^1 c_Y^{-1}(x,y)p(x,y)^2\, dy < \infty$. Since the numerator of \eqref{eqn:integrability_assum} is bounded in $y$ for every fixed $x$ it follows that Assumption~\ref{ass:nondegeneracy}\ref{item:assum_new_3} is also satisfied.

 For this specification we have that \[\ell_Y(x,y) = \frac{\alpha_1 + \beta_1-1}{2y} - \frac{\alpha_2 + \beta_2 -1}{2(1-y)}\] so that under the parameter conditions \eqref{eqn:param_constraints} and \eqref{eqn:beta_bounds} we have that $\int_0^1 \ell_Y^2 c_Y p < \infty$. 
 
 Hence it only remains to verify Assumption~\ref{ass:nondegeneracy}\ref{item:test_func}. Note that in this setup $A(x) = x^{a_1 + b_1 -1}(1-x)^{a_2 + b_2 -1}h(x)$, where $h(x)$ is positive and bounded. Similarly we have that $B(y) = Cy^{\alpha_1 + \beta_1-1}(1-y)^{\alpha_2 + \beta_2 -1}$ for some constant $C > 0$. For $r \in \R$ define $u_n(r) = n(r-1/n)\land  n(1 - 1/n - r) \land 1 \lor 0$. Let $\eta_n\in C_c^\infty(\R)$ be nonnegative with $\int_0^1 \eta(r)\, dr = 1$ and such that $\mathrm{supp}(\eta_n) \subset (-1/(2n),1/(2n))$. Define
 \[\varphi_n := \psi_n := (u_n * \eta_n)|_{(0,1)}.\]
By construction $\varphi_n(x) = 0$ whenever $x \leq 1/(2n)$ or $x \geq 1-1/(2n)$ and $\varphi_n(x) = 1$ whenever $\frac{5}{2n} \leq x \leq 1-\frac{5}{2n}.$ In particular we see that $\varphi_n \to 1$ pointwise. 
Finally note that $u_n$ is $n$-Lipschitz, so by properties of convolution $\varphi_n$ is also $n$-Lipschitz. Since $\varphi_n$ is only nonconstant for $x \in (1/(2n),5/(2n)) \cup (1-5/(2n), 1- 1/(2n))$ we have from this and the Lipschitz property that 
\begin{align*}|\varphi_n'(x)| \leq n & \leq \frac52 \max\left\{\frac{1}{x}, \frac{1}{1-x}\right\}(1_{(\frac{1}{2n}, \frac{5}{2n})}(x) + 1_{(1-\frac{5}{2n}, 1-\frac{1}{2n})}(x)) \\
& \leq \frac52\frac{1}{x(1-x)}(1_{(\frac{1}{2n}, \frac{5}{2n})}(x) + 1_{(1-\frac{5}{2n}, 1-\frac{1}{2n})}(x)) 
\end{align*}
Consequently we see that 
\[\lim_{n \to \infty} \int_0^1 \varphi_n'(x)^2 A(x)\, dx \leq  \lim_{n \to \infty} \int_0^1 x^{a_1 + b_1 -3}(1-x)^{a_2 + b_2 -3}(1_{(\frac{1}{2n}, \frac{5}{2n})}(x) + 1_{(1-\frac{5}{2n}, 1-\frac{1}{2n})}(x))\, dx = 0, \]
where the last equality followed by the dominated convergence theorem since \[\int_0^1 x^{a_1+b_1-3}(1-x)^{a_2 + b_2-3}\, dx < \infty\] due to our choice of parameters \eqref{eqn:param_constraints}. A similar calculation shows that $\lim_{n \to \infty} \int_0^1 \psi_n'(y)^2B(y)\, dy = 0$.
We have now verified Assumptions~\ref{ass:regularity} and \ref{ass:nondegeneracy} so that Theorem~\ref{thm:main} applies.

When $q_1 = q_2 = 1$, direct calculations show that \begin{align*} \hat \phi(x) &=  \frac{1}{2}\log\left(\int_0^1 c_Xp(x,y)dy\right) = \frac{1}{2}\log\left(\frac{x^{a_1+b_1-1}(1-x)^{a_2+b_2-1}}{B(a_1,a_2)}\right)\\
	&\quad + \frac{1}{2}\log\left(x(1-x) +x\frac{B(\alpha_1,\alpha_2+1)}{B(\alpha_1,\alpha_2)} + (1-x)\frac{B(\alpha_1+1,\alpha_2)}{B(\alpha_1,\alpha_2)} + \frac{B(\alpha_1+1,\alpha_2+1)}{B(\alpha_1,\alpha_2)}\right)
\end{align*} and the robust growth-optimal strategy is given by \begin{align*} \hat \theta_t &=  \frac{a_1 + b_1 -1}{2X_t} - \frac{a_2 + b_2 -1}{2(1-X_t)}  \\
&\quad + \frac{1}{2} \frac{1-2X_t+ \frac{B(\alpha_1,\alpha_2+1)}{B(\alpha_1,\alpha_2)} -\frac{B(\alpha_1+1,\alpha_2)}{B(\alpha_1,\alpha_2)} }{X_t(1-X_t) +X_t\frac{B(\alpha_1,\alpha_2+1)}{B(\alpha_1,\alpha_2)} + (1-X_t)\frac{B(\alpha_1+1,\alpha_2)}{B(\alpha_1,\alpha_2)} + \frac{B(\alpha_1+1,\alpha_2+1)}{B(\alpha_1,\alpha_2)}}.
\end{align*} 

Finally note that the for the given inputs $c_X$ and $p$, Theorem~\ref{thm:worst_case_measure_existence} yields a worst case measure for every $\beta_1,\beta_2$ satisfying \eqref{eqn:beta_bounds}. Indeed, $\beta_1$ and $\beta_2$ affect the dynamics of $Y$, but do not enter in the dynamics of $X$ or in $\hat \phi$ so that $\hat \theta$ is growth-optimal under all of the corresponding measures. In particular, we see that there are uncountably many worst-case measures in this case.



	\appendix 
	
	\section{Estimates for certain degenerate elliptic PDEs and variational problems} \label{sec:PDE}
	Fix $m \in \N$ and a bounded convex domain $D \subset \R^m$. 
	In this appendix we collect and prove some results for the minimizer of the variational problem
	\begin{equation} \label{variational_problem} \min_{u \in \Wcal_0} J(u),
	\end{equation}
	where $\Wcal_0$ will be specified below and 
	\[J(u) = \frac12\int_D (\nabla u - \xi) ^\top a (\nabla u - \xi) + \int_D fu.\]
	Here $a:D \to \mathbb{S}^m_{++}$,
	$\xi: D \to \R^m$, and $f:D \to \R$ are measurable and satisfy Assumption~\ref{ass:pde} below. We present the results of this section in a general setting and then apply them in Appendix~\ref{sec:proofs}, where they play a crucial role in the proof of Lemma~\ref{lem:hat_v}. We now introduce some notation used in both appendices.
	\paragraph{Notation}
	\begin{itemize}[noitemsep]
	    \item $L^p_w(D)$ for measurable $w:D \to (0,\infty)$ is the weighted $L^p$ space with norm $\|u\|_{L^p_w(D)} = (\int_D |u|^p\, w)^{1/p}$,
	    \item The mean over $D$ with weight $w$ of any $f \in L^1_w(D)$ is denoted by $f_{w,D} = \int_D fw/\int_D w$. If $w \equiv 1$ we write $f_D$ for $f_{1,D}$,
	    \item $U \Subset V$: the closure $\overline U$ is compact and contained in $V$,
	    \item $\delta_U(y) = \inf_{x\in U} |x-y|$: distance from $y$ to the set $U$.
	\end{itemize}
	
	We make the following additional assumptions on the coefficients.
	\begin{assum}[Coefficients assumption] \label{ass:pde} Set $w(y) = \lambda_{\min}(a(y))$ for $y \in D$. We assume that
\begin{enumerate}[noitemsep]
	\item \label{item:cont} $a$ is locally Lipschitz continuous on $D$,
    \item \label{item:weight} $w(y) = \rho(y)^k$ for some positive concave function $\rho$ and some $k \geq 0$,
    \item $\int_D \xi^\top a \xi < \infty$,
    \item\label{item:zero_mean} $f_D = 0$ and $f/\sqrt{w} \in L^2(D)$.
\end{enumerate}	    
	\end{assum}
	As a consequence of Assumption~\ref{ass:pde}\ref{item:cont} we have that $a$ is uniformly elliptic on compact subsets of $D$, but may degenerate at the boundary.
	The Euler--Lagrange PDE associated to \eqref{variational_problem} is
\begin{equation} \label{PDE_div}
	\div(a (\nabla u - \xi)) = f.
\end{equation}
For $u \in W^{1,2}_\mathrm{loc}(D)$ the standard weak formulation of \eqref{PDE_div} is
\begin{equation} \label{PDE_weak}
	\int_D \left( \nabla \varphi^\top a (\nabla u - \xi) + \varphi f \right) = 0 \text{ for all } \varphi \in C^1_c(D).
\end{equation}
Solutions $u$ will be found in a space $\Wcal_0 \subset W^{1,2}_{\mathrm{loc}}(D)$ which we now introduce.
We rely on the following weighted Poincar\'e inequality of Chua and Wheeden \cite[Theorem~1.1]{chua2006estimates}: for any bounded convex domain $U \subset D$ we have that
\begin{equation} \label{eqn:weighted_poincare_lipschitz}
	\|u - u_{w,U}\|_{L^2_w(U)} \leq \frac{\mathrm{Diam}(U)}{\pi}\|\nabla u\|_{L^2_w(U)}
\end{equation}
for every Lipschitz function $u$ on $U$. Now let $\{D_n\}_{n \in \N}$ be a sequence of bounded convex domains such that $D_n \Subset D_{n+1}$ and $\cup_n D_n = D$. By the density of Lipschitz functions in $W^{1,2}(D_n)$ and the uniform ellipticity of $a$ on $D_n$ we obtain \eqref{eqn:weighted_poincare_lipschitz} for all $u \in W^{1,2}(D_n)$ by approximation.

Set
\begin{equation} \label{eq_norm_0}
	\| u \|_\Wcal = \left( \int_D \nabla u^\top a \nabla u \right)^{1/2}
\end{equation}
for any weakly differentiable $u$, and define
\[
\Wcal = \left\{u \in L^2_w(D) \cap W^{1,2}_\mathrm{loc}(D)  \colon \| u \|_\Wcal < \infty \text{ and } \int_D uw = 0 \right\}.
\]
Using \eqref{eqn:weighted_poincare_lipschitz} and the inequality $\int_D |\nabla u|^2 w \leq \int_D \nabla u^\top a \nabla u$ one has
\[
	\| u - u_{w,D_n}\|_{L^2_w(D_n)} \le \frac{\mathrm{Diam}(D_n)}{\pi} \|\nabla u\|_{L^2_w(D_n)} \leq \frac{\mathrm{Diam}(D)}{\pi} \|\nabla u\|_{L^2_w(D)} \leq \frac{\mathrm{Diam}(D)}{\pi} \|u\|_{\Wcal} 
\] 
for $u \in \Wcal$.
Sending $n \to \infty$,  using dominated convergence and the fact that $u_{w,D} = 0$ for $u \in \Wcal$ yields 
\begin{equation} \label{eq_W_norm_bound}
	\|u\|_{L^2_w(D)} \leq \frac{\mathrm{Diam}(D)}{\pi} \|u\|_{\Wcal}.
\end{equation}
 Next note that for $v \in \Wcal$ we have, courtesy of \eqref{eq_W_norm_bound}, the bound
\begin{equation} \label{functional_bound} 
\left|\int_D fv \right| = \left|\int_D \frac{f}{\sqrt{w}}v\sqrt{w}\right| \leq \left\|\frac{f}{\sqrt{w}}\right\|_{L^2(D)}\|v\|_{L^2_w(D)} \leq \frac{\mathrm{Diam}(D)}{\pi} \left\|\frac{f}{\sqrt{w}}\right\|_{L^2(D)}\|v\|_{\Wcal}
\end{equation}

 The norm $\|\cdot\|_{\Wcal}$ is induced by the inner product $(u,v)_{\Wcal} = \int_D \nabla u^\top a \nabla v$. Equipped with this inner product $\Wcal$ becomes a Hilbert space, since any Cauchy sequence in $\Wcal$ converges in $L^2_w(D)$ and in $W^{1,2}(U)$ for every $U \Subset D$. Indeed, the $L^2_w(D)$ convergence follows from \eqref{eq_W_norm_bound}. To see the $W^{1,2}(U)$ convergence first note that $\|\cdot\|_{L^2_w(U)}$ is equivalent to $\|\cdot\|_{L^2(U)}$  since $w$ is bounded from above and away from zero on $U$. As such, there exists a $\kappa_U > 0$ such that
 \[ \|u\|_{W^{1,2}(U)} \leq \kappa_U (\|u\|_{L^2_w(U)} + \|\nabla u\|_{L^2_w(U)})\leq \kappa_U(1+\frac{\mathrm{Diam}(D)}{\pi})\|u\|_{\Wcal},\]
 where the last inequality follows from the definition of $w$ and \eqref{eq_W_norm_bound}. Define the subspace
\[
\Wcal_0 = \overline{ \{\varphi - \varphi_{w,D} \colon \varphi \in C^1_c(D) \} }^{\| \cdot \|_\Wcal}.
\]
This is where we will look for solutions to \eqref{variational_problem}. We are now ready to state the main result of this appendix. To simplify the notation to come write $|\xi|_{\Wcal}^2$ for $\int_D \xi^\top a \xi$.

\begin{thm}[Characterization of minimizer] \label{T_pde_no_param}
There exists a unique solution $\hat u \in \Wcal_0$ of \eqref{variational_problem}, and this solution satisfies
	\begin{equation} \label{PDE_norm_bound_1}
		\| \hat u \|_\Wcal \le \frac{2 \mathrm{Diam}(D)}{\pi} \left\| \frac{f}{\sqrt{w}} \right\|_{L^2(D)} + 2|\xi|_{\Wcal}
	\end{equation}
	Moreover, $\hat u$ is the unique solution in $\Wcal_0$ of \eqref{PDE_weak}.
 If in addition $a \in C^{1,\alpha}(D)$ and $f \in C^\alpha(D)$ for some $\alpha \in (0,1)$, then $\hat u$ belongs to $C^{2,\alpha}(D)$ and satisfies \eqref{PDE_div} classically.
\end{thm}

The proof of Theorem~\ref{T_pde_no_param} is broken up into several lemmas. However, we first prove another result, which is a consequence of Theorem~\ref{T_pde_no_param}. For a set $U \subset \R^m$ and $\varepsilon > 0$, the open $\varepsilon$-fattening of $U$ is the set $U_\varepsilon$ of all points closer than $\varepsilon$ to $U$, namely $U_\varepsilon = \{x \in \R^m \colon \delta_U(x) < \varepsilon\}$.

\begin{thm} \label{T_regularity_eq}
Fix $q \in [2,\infty)$, suppose that $f \in L^q_{\mathrm{loc}}(D)$ and $a\xi \in W^{1,q}_{\mathrm{loc}}(D;\R^m)$.
	 Fix $U \Subset D$ and let $\varepsilon$ be a positive number less than the distance between $U$ and $\partial D$, i.e., $\varepsilon \in (0, \inf_{y \in U} \delta_{\partial D}(y))$. Let $\kappa \in (0,1)$ be such that
	\begin{equation} \label{eqn:quant_bound}
	\inf_{y \in U_\varepsilon} \lambda_\textnormal{min}(a(y)) \ge \kappa, \quad \| a \|_{L^\infty(U_\varepsilon)} \le \kappa^{-1}, \quad \| \div(a) \|_{L^\infty(U_\varepsilon)} \le \kappa^{-1}.
	\end{equation} 
	Then the unique solution $\hat u \in \Wcal_0$ of \eqref{variational_problem} satisfies
	\begin{equation} \label{eqn:lp_loc_bound}
\| \hat u \|_{W^{2,q}(U)} \le C \left( \| f \|_{L^q(U_\varepsilon)} + \|\div(a\xi)\|_{L^q(U_\varepsilon)} +\left\|\frac{f}{\sqrt{w}}\right\|_{L^2(D)} + |\xi|_{\Wcal} \right),
	\end{equation}
	where $C$ is a constant that only depends on $m$, $q$, $\varepsilon$, $\kappa$, the volume of $U$, the modulus of continuity of $a$ on $U_\varepsilon$, and $\mathrm{Diam}(D)$.
\end{thm}

\begin{proof}
Note that Assumption~\ref{ass:pde} ensures that $\kappa$ as in \eqref{eqn:quant_bound} can be found.	We use the notation $B_r(z)$ for the open ball of radius $r$ centered at $z$. Select points $z_1,\ldots,z_n \in U$ such that the balls $B_{\varepsilon/2}(z_i)$ cover $U$. The number of points required, $n$, can be bounded in terms of $\varepsilon$ and the volume of $U$. Fix $i$. Since $B_\varepsilon(z_i) \Subset D$, we have from \cite[Theorem~8.8]{gilb1988elliptic} that $\hat u \in W^{2,2}(B_\varepsilon(z_i))$ and that
	\[
	\Tr(a \nabla^2 \hat u) + \div(a)^\top \nabla \hat u = f + \div(a\xi) \text{ a.e. in } B_\varepsilon(z_i).
	\]
	From \cite[Theorem~11.2.3]{krylov2008lectures} with the operator $L = \Tr(a \nabla^2) + \div(a)^\top \nabla$ we get
	\begin{equation} \label{T_regularity_eq1}
		\| \hat u \|_{W^{2,q}(B_{\varepsilon/2}(z_i))} \le C' \left( \| f \|_{L^q(B_\varepsilon(z_i))} + \|\div(a\xi)\|_{L^q(B_\varepsilon(z_i))}+ \| \hat u \|_{L^2(B_\varepsilon(z_i))} \right)
	\end{equation}
	where the constant $C'$ only depends on $d,q,\varepsilon,\kappa$, and the modulus of continuity of $a$ on $B_\varepsilon(z_i)$, hence on $U_\varepsilon$. (More specifically, we apply that theorem with $\Omega = B_\varepsilon(z_i)$, $R = \varepsilon/2$, $z = z_i$, and $(p,q)$ replaced by $(2,q)$. We also modify the coefficients of $L$ outside $B_\varepsilon(z_i)$ so that the assumed bounds and modulus of continuity hold globally.) Recall that $U \subset \bigcup_{i=1}^n B_{\varepsilon/2}(z_i)$. Thus by summing \eqref{T_regularity_eq1} over $i$ and using that $B_\varepsilon(z_i) \subset U_\varepsilon$ for all $i$ we obtain
	\[
	\| \hat u \|_{W^{2,q}(U)} \le \sum_{i=1}^n \| \hat u \|_{W^{2,q}(B_{\varepsilon/2}(z_i))} \le n C' \left( \| f \|_{L^q(U_\varepsilon)} + \|\div(a\xi)\|_{L^q(U_\varepsilon)} + \| \hat u \|_{L^2(U_\varepsilon)} \right).
	\]
	Next note that $w$ is bounded away from zero on $U_\varepsilon$ so there exists a $\kappa_\varepsilon \geq 1$ such that $\|\hat u\|_{L^2(U_\varepsilon)} \leq \kappa_\varepsilon \|\hat u\|_{L^2_w(U_\varepsilon)} \leq \kappa_\varepsilon \|\hat u\|_{L^2_w(D)}$, where the last inequality follows since $U_\varepsilon \subset D$.
	Thanks to \eqref{eq_W_norm_bound} and \eqref{PDE_norm_bound_1} we have $\| \hat u \|_{L^2_w(D)} \le\frac{\mathrm{Diam}(D)}{\pi} \| f/\sqrt{w} \|_{L^2(D)} + 2|\xi|_{\Wcal}$ so the result follows with the constant $C = nC'\kappa_\varepsilon(\frac{\mathrm{Diam}(D)}{\pi} + 2)$.
\end{proof}

We now turn our attention to proving Theorem~\ref{T_pde_no_param}.
Recall the variational problem \eqref{variational_problem}.
Note that $J$ is strictly convex on $\Wcal_0$, continuous in the norm topology on $\Wcal_0$, and lower semicontinuous in the weak topology on $\Wcal_0$. This follows from the corresponding convexity and continuity properties of norms and the fact that the bounded linear functional $v \mapsto \int_D fv$ is both strongly and weakly continuous on $\Wcal_0$.

\begin{lem} \label{L_FOC}
	An element $u \in \Wcal_0$ is optimal for \eqref{variational_problem} if and only if it satisfies \eqref{PDE_weak}.
\end{lem}

\begin{proof}
	Suppose $u \in \Wcal_0$ is optimal. For any $\varepsilon >0$ and $\varphi \in C_c^1(D)$ one has
	\[
	0 \le \frac{1}{\varepsilon} \left( J(u + \varepsilon (\varphi - \varphi_{w,D})) - J(u) \right) = \int_D \nabla \varphi^\top a (\nabla u - \xi) + \int_D f\varphi + \frac{\varepsilon}{2} \int_D \nabla \varphi^\top a \nabla \varphi,
	\]
	where we used the fact that $f_D = 0$.
	Applying this with $\pm \varphi$ and sending $\varepsilon \to 0$ yields \eqref{PDE_weak}. Conversely, suppose $u \in \Wcal_0$ is not optimal, so that $J(u+v) - J(u) < 0$ for some $v \in \Wcal_0$. By density in $\Wcal_0$ we can assume that  $v = \varphi - \varphi_{w,D}$ for some $\varphi \in C_c^1(D)$. Then
	\[
	0 > J(u+v) - J(u) = \int_D \nabla \varphi^\top a (\nabla u - \xi) + \int_D f\varphi + \frac12 \int_D \nabla \varphi^\top a \nabla \varphi \ge \int_D \nabla \varphi^\top a (\nabla u - \xi) + \int_D f\varphi
	\]
	showing that $u$ does not satisfy \eqref{PDE_weak}.
\end{proof}

\begin{remark} \label{rem:dense} 
Since $f_D =0$, by density of $\{\varphi-\varphi_{w,D}: \varphi \in C_c^1(D)\}$ in $\Wcal_0$ we can equivalently require that \eqref{PDE_weak} holds for all $\varphi \in \Wcal_0$.
\end{remark}

In view of this identification, the existence and uniqueness statements in Theorem~\ref{T_pde_no_param} follow once we prove the corresponding properties for the variational problem \eqref{variational_problem} and its solution.

\begin{lem} \label{L_var_problem_solution}
	The variational problem \eqref{variational_problem} has a unique optimal solution $\hat u \in \Wcal_0$. This solution satisfies \eqref{PDE_norm_bound_1}.
\end{lem}

\begin{proof}
	Uniqueness follows from strict convexity of $J$. To prove existence of $\hat u$ observe courtesy of \eqref{functional_bound} that \[J(u) \ge \frac12 \|u\|_\Wcal^2 - \|u\|_\Wcal \left(\frac{\mathrm{Diam}(D)}{\pi}\left\|\frac{f}{\sqrt{w}}\right\|_{L^2(D)} + |\xi|_{\Wcal}\right) + \frac{1}{2}|\xi|^2_{\Wcal}\] for any $u \in \Wcal_0$. This can be rearranged to
	\begin{equation} \label{eq_J_bounds}
		\| u \|_\Wcal \le \frac{\mathrm{Diam}(D)}{\pi}\left\|\frac{f}{\sqrt{w}}\right\|_{L^2(D)} + |\xi|_{\Wcal} + \sqrt{\left(\frac{\mathrm{Diam}(D)}{\pi}\left\|\frac{f}{\sqrt{w}}\right\|_{L^2(D)} + |\xi|_{\Wcal}\right)^2 + 2J(u) - |\xi|_{\Wcal}^2}.
	\end{equation}
	Write $\hat J = \inf_{u \in \Wcal_0} J(u)$ and consider a sequence $(u_n)_{n \in \N} \subset \Wcal_0$ such that $J(u_n) \to \hat J$. Since $J(0) < \infty$ we clearly have that $\hat J < \infty$. Thanks to \eqref{eq_J_bounds}, $(u_n)_{n \in \N}$ is bounded in the Hilbert space $\Wcal_0$ and hence admits a weakly convergent subsequence, again denoted by $(u_n)_{n \in \N}$. Call the limit $\hat u$. Weak lower semicontinuity yields $J(\hat u) \le \liminf_n J(u_n) = \hat J$, so $\hat u$ is optimal. The estimate \eqref{PDE_norm_bound_1} follows from \eqref{eq_J_bounds} and the fact that $J(\hat u) = \hat J \le J(0) = \frac{1}{2}|\xi|_{\Wcal}^2$.
\end{proof}
The last statement of Theorem~\ref{T_pde_no_param} is a consequence of the following two lemmas, where $\hat u$ is the unique optimal solution of \eqref{variational_problem}, and hence the unique solution of \eqref{PDE_weak}, obtained in Lemma~\ref{L_var_problem_solution}.

\begin{lem} \label{L_Holder_continuous_1}
	If $a\xi \in L^{q/2}_{\mathrm{loc}}(D)$ and $f \in L^q_\mathrm{loc}(D)$ for some $q > d$ then $\hat u$ is locally H\"older continuous in $D$.
\end{lem}

\begin{proof}
	Let $U \Subset D$. By Assumption~\ref{ass:pde}\ref{item:cont} we have that $a, \div(a) \in L^\infty(U)$ and $a$ is uniformly elliptic in $U$. The hypotheses of the lemma ensure that  $a\xi \in L^{q/2}(U)$ and $f \in L^q(U)$. Additionally, $\hat u$ belongs to $W^{1,2}(U)$ and $\hat u$ is a weak solution of \eqref{PDE_div} in $U$ in the sense that \eqref{PDE_weak} holds with $U$ in place of $D$. Thus \cite[Theorem~8.22]{gilb1988elliptic} implies that $\hat u$ is locally H\"older continuous in $U$, and hence in $D$ since $U$ was arbitrary.
\end{proof}

\begin{lem} \label{lem:smooth_regularity}
	Assume that $a, \xi \in C^{1,\alpha}(D)$, $f \in C^\alpha(D)$, and $f_D = 0$ for some $\alpha \in (0,1)$. Then $\hat u \in C^{2,\alpha}(D)$ and satisfies \eqref{PDE_div} classically in $D$. 
\end{lem}

\begin{proof}
	Consider an open ball $U \Subset D$. We have that $a$ is uniformly elliptic and uniformly Lipschitz in $U$, $f + \div(a\xi)$ belongs to $L^2(U)$, and $\hat u \in W^{1,2}(U)$ is a weak solution of \eqref{PDE_div} in $U$ in the sense that \eqref{PDE_weak} holds with $U$ in place of $D$. Thus \cite[Theorem~8.8]{gilb1988elliptic} implies that $\hat u$ belongs to $W^{2,2}_{\mathrm{loc}}(U)$ and satisfies the PDE \eqref{PDE_div} in non-divergence form,
	\begin{equation} \label{L_Holder_continuous_2_eq0}
		\Tr(a\nabla^2 u) + \div(a)\cdot\nabla u = f + \div(a\xi),
	\end{equation}
	almost everywhere in $U$. 
	
	We claim that, in fact, $\hat u \in W^{2,p}_{\mathrm{loc}}(U)$ for all $p \in [1,\infty)$. To see this, consider any open ball $V \Subset U$ and let $p \ge 2$ (this suffices). Since $\hat u$ belongs to $W^{2,2}(V)$ and satisfies \eqref{L_Holder_continuous_2_eq0} almost everywhere in $V$, and since $f \in L^p(V)$ for all $p$, it follows from \cite[Theorem~11.2.3]{krylov2008lectures} that $\hat u$ belongs to $W^{2,p}_\mathrm{loc}(V)$ for all $p$. (Specifically, in that theorem we take $\Omega = V$, $z$ the center of $V$ and $R$ half its radius, and $(p,q)$ replaced by $(2,p)$.) Since $V$ was arbitrary, we conclude that $\hat u \in W^{2,p}_\mathrm{loc}(U)$ for all $p$ as claimed.
	
	Next, we consider the PDE \eqref{L_Holder_continuous_2_eq0} in $U$ with boundary condition $u = \hat u$ on $\partial U$. The hypotheses imply that $a$ is uniformly elliptic in $U$ and that $a$, $\div(a)$, $\div(a\xi)$ and $f$ belong to $C^\alpha(U)$. Thanks to Lemma~\ref{L_Holder_continuous_1}, $\hat u$ is continuous on $\partial U$. Because $U$ is a ball it satisfies an exterior sphere condition at every boundary point. \cite[Theorem~6.13]{gilb1988elliptic} then implies that there is a unique classical solution $\hat v \in C(\overline{U}) \cap C^{2,\alpha}(U)$.
	
	To summarize, both $\hat u$ and $\hat v$ belong to $W^{2,p}_\mathrm{loc}(U) \cap C( \overline U)$ for all $p \in [1,\infty)$ and satisfy \eqref{L_Holder_continuous_2_eq0} almost everywhere in $U$ as well as $\hat u = \hat v$ on $\partial U$. Taking $p=d$ we may apply \cite[Theorem~9.5]{gilb1988elliptic} to conclude that $\hat u = \hat v$. Thus $\hat u \in C^{2,\alpha}(U)$ and satisfies \eqref{L_Holder_continuous_2_eq0}, hence \eqref{PDE_div}, classically in $U$. Since the ball $U$ was arbitrary, the lemma is proved.
\end{proof}

\begin{proof}[Proof of Theorem~\ref{T_pde_no_param}]
	The result follows from Lemmas~\ref{L_FOC}, \ref{L_var_problem_solution} and \ref{lem:smooth_regularity}.
\end{proof}

\section{A measurability result} \label{S_measurability}

Returning to our original problem we consider now a parameter dependent version of the setting of Appendix~\ref{sec:PDE}, and assume that the functions $a, f, \xi$ are indexed by a parameter $x$ from an open set $E \subset \R^d$. We indicate this by writing $a^x, f^x, \xi^x$. The objective function $J$ then also depends on $x$ and is given by
\[
J^x(u) = \frac{1}{2} \int_D (\nabla u - \xi^x)^\top a^x (\nabla u - \xi^x) + \int_D f^x u.
\]
We suppose throughout this section that Assumption~\ref{ass:pde} is satisfied for each fixed $x$, and we write $w^x(y) = \lambda_\text{min}(a^x(y))$ for the weight function. We then obtain $x$-dependent spaces $\Wcal^x$, $\Wcal_0^x$, and norm $\|\cdot\|_{\Wcal^x}$. From Theorem~\ref{T_pde_no_param} we get, among other things, that there exists a unique $\hat u^x \in \Wcal_0^x$ that minimizes $J^x$ over $\Wcal_0^x$, for each $x$. Our goal is to prove the following regularity result.

\begin{thm} \label{T_measurable}
In addition to the above, assume that
\begin{equation} \label{T_jointly_mb_eq1}
a^x(y), f^x(y), \xi^x(y) \text{ are jointly continuous in } (x,y) \in E \times D.
\end{equation}
Then there exists a Borel measurable function $\hat v \colon E \times D \to \R$ such that for a.e.\ $x \in E$, the function $\hat v^x = \hat v(x, \cdot)$ is a version of $\hat u^x$.
\end{thm}

\begin{proof}
\textit{Step 1.} For each $N \in \N$, consider the set of functions in $C^1_c(D)$ that vanish everywhere within distance $1/N$ of the boundary of $D$ and whose gradients are $N$-Lipschitz. Note that any such function is itself bounded by $N \text{Diam}(D)$. Let $U_N$ denote the norm closure of this set in $W^{1,2}(D)$. The following properties are easily established:
\begin{enumerate}
\item\label{mb_pf_UN_prop_0} $U_N$ is convex.
\item\label{mb_pf_UN_prop_1} $U_N$ is compact. Indeed, the set $\{\nabla \varphi\colon \varphi \in U_N\}$ is equicontinuous and uniformly bounded, so the theorem of Arzel\`a--Ascoli implies that $U_N$ is even relatively compact, whence compact, in the $C^1$ topology.
\item\label{mb_pf_UN_prop_2} Every function in $U_N$ has a continuous version, and we always use this version.
\item\label{mb_pf_UN_prop_3} Every function in $C^1_c(D)$ belongs to $U_N$ for some $N$.
\item\label{mb_pf_UN_prop_4} For every $u \in U_N$ and $x \in E$, the centered function $u - u_{w^x,D}$ belongs to $\Wcal_0^x$.
\end{enumerate}
For each $N$ and $x$, we also consider the set $V_N^x$ of centered functions,
\[
V_N^x = \{u - u_{w^x,D} \colon u \in U_N\}.
\]
Thanks to \ref{mb_pf_UN_prop_3}--\ref{mb_pf_UN_prop_4} above, the union $\bigcup_{N \in \N} V_N^x$ is a dense subset of $\Wcal_0^x$ for every $x$.

\textit{Step 2.} Fix $N$. For each $x$, let $\hat u_N^x$ be the solution to the minimization problem
\[
\inf_{u \in U_N} J^x(u).
\]
The minimizer exists and is unique because $J^x$ is continuous and strictly convex on the compact convex set $U_N$ with respect to the $W^{1,2}(D)$ norm. Next, the hypothesis \eqref{T_jointly_mb_eq1} implies that $J^x(u)$ is measurable in $x$ for each fixed $u \in U_N$. Thus $(x,u) \mapsto J^x(u)$ is a Carath\'eodory function on $E \times U_N$. The measurable maximum theorem \cite[Theorem~18.19]{alip2006infinite} then yields that $x \mapsto \hat u_N^x$ is measurable from $D$ to $U_N$. Furthermore, the evaluation map $(u,y) \mapsto u(y)$ is measurable on $U_N \times D$, since it is continuous. Indeed, if $u_n, u$ are in $U_N$ with $u_n \to u$ in $L^2_\text{loc}(D)$, and $y_n$, $y$ are in $D$ with $y_n \to y$, then thanks to the uniform Lipschitz constant $N$, we have $u_n(y_n) \to u(y)$.

\textit{Step 3.} Next, we center the minimizers $\hat u_N^x$. Specifically, we define
\[
\hat v_N^x = \hat u_N^x - (\hat u_N^x)_{w^x,D} \in V_N^x.
\]
Thanks to Assumption~\ref{ass:pde}\ref{item:zero_mean} $J^x$ is unaffected by constant shifts in the sense that $J^x(u) = J^x(u + c)$ for any $u \in U_N$ and $c \in \R$. It follows that $\hat v_N^x$ minimizes $J^x$ over $V_N^x$. Moreover, Fubini's theorem implies that $x \mapsto (\hat u_N^x)_{w^x,D}$ is measurable, so $(x,y) \mapsto \hat v_N^x(y)$ inherits joint measurability from $(x,y) \mapsto \hat u_N^x(y)$.

We can now pass to the limit as $N \to \infty$. For each fixed $x$, the sequence $(\hat v_N^x)_{N \in \N}$ is bounded in $\Wcal_0^x$ thanks to \eqref{eq_J_bounds}, so has a weakly convergent subsequence. We denote the weak limit of this subsequence by $\hat w^x$. By weak lower semicontinuity, $J^x(\hat w^x) \le \liminf_{N \to \infty} J^x(\hat v_N^x) = \inf_v J^x(v)$, where the infimum is taken over the union $\bigcup_{N \in \N} V_N(x)$. As remarked above, this union is dense in $\Wcal_0^x$. We conclude that, as an element of $\Wcal_0^x$, $\hat w^x$ is equal to $\hat u^x$, the unique minimizer of $J^x$ over $\Wcal_0^x$. Since this holds for any subsequence of $(\hat v_N^x)_{N \in \N}$, the sequence actually converges weakly in $\Wcal_0^x$ to $\hat u^x$.

\textit{Step 4.} We are now in the position to select versions of $\hat u^x$ that ensure joint measurability. The functions $\hat v_N \colon (x,y) \mapsto \hat v_N^x(y)$ are jointly measurable on $E \times D$. Consider an increasing sequence of open sets $U_k \Subset E \times D$, $k \in \N$, whose union is $E \times D$. Thanks to \eqref{eq_J_bounds} and the joint continuity Hypothesis \eqref{T_jointly_mb_eq1}, the sequence $(\hat v_N)_{N \in \N}$ is bounded in $L^1(U_k)$ for each $k$. We may then apply Koml\'os lemma on each $U_k$ together with a diagonal argument to get a sequence $(\tilde{v}_N)_{N \in \N}$ of forward convex combinations of the $\hat v_N$ that converges a.e.\ to some measurable limit $\hat{v}$.

Now, for every $x$ in some full-measure set $E' \subset E$, $\tilde{v}_N^x$ converges to $\hat{v}^x$ a.e. Moreover, for every $x$, the sequence $\tilde{v}_N^x$ still converges weakly to $\hat u^x$. It follows that, for every $x \in E'$, $\hat{v}^x$ is a version of $\hat u^x$. This completes the proof.
\end{proof}

    \section{Proofs for Section~\ref{sec:results}} \label{sec:proofs}
	The purpose of this section is to prove the remaining results of Section~\ref{sec:results}.
 
    \subsection{Proofs for Section~\ref{sec:main_result}}
    
    We start with Proposition~\ref{prop:hat_phi}.
	\begin{proof}[Proof of Proposition~\ref{prop:hat_phi}]
	The existence and regularity of $\hat \phi$ as well as \eqref{eqn:phi_pde} will follow from \cite[Lemma~A.1]{kardaras2021ergodic} once we establish that $\int_E \div A^\top A^{-1} \div A < \infty$. To this end for a measurable vector field $\Psi:E \to \R^d$ define the quantity
	\[H(\Psi) = \int_F (\ell_X(x,y) - \Psi(x))^\top c_X(x,y)(\ell_X(x,y) - \Psi(x))p(x,y)\, dx\, dy.\]
	Clearly $H$ is nonnegative and by choosing $\Psi = 0$ we get the bound $0 \leq \inf_{\Psi} H(\Psi) \leq \int_F \ell_X^\top c_X \ell_X p < \infty$, where finiteness is due to Assumption~\ref{ass:nondegeneracy}\ref{item:technical_1}. A direct calculation integrating out the $y$ component shows that
	\[H(\Psi) = \int_E (\Psi^\top A \Psi - 2 \div A^\top \Psi) + \int_F \ell_X^\top c_X \ell_X p.\]
	Minimizing the integrand pointwise yields 
	\[\inf_\Psi H(\Psi) = H(A^{-1}\div A) = - \int_E \div A^\top A^{-1}\div A + \int_F \ell_X^\top c_X \ell_X\, p.\]
	Since this is nonnegative we obtain $\int_E \div A^\top A^{-1}\div A \leq \int_F \ell_X^\top c_X \ell_X p < \infty$.
	
	Thus is just remains to prove that $\hat \phi$ minimizes \eqref{eqn:argmin}. To this end recall the measure $\P^0$ constructed in Section~\ref{sec:assumptions}, where we omit the starting point for notational convenience. Using the dynamics \eqref{eqn:tilde_p_dynamics} and arguing in the exact same way as in \cite[Lemma~3.4]{itkin2020robust} yields that for any portfolio $\theta_t = \theta(X_t,Y_t)$ in feedback form we have 
	\begin{equation} \label{eqn:P0_growth} g(\theta;\P^0) = \frac12\int_F\left(\ell_X^\top c_X \ell_X -(\ell_X - \theta)^\top c_X(\ell_X - \theta)\right)p.
	\end{equation} 
    Indeed, the proof of \cite[Lemma~3.4]{itkin2020robust} only relies on the optimal growth rate being finite and the ergodic property, both of which hold under $\P^0$. Now let $\phi \in \Dcal$ be given and consider the trading strategy $\theta^\phi_t = \nabla \phi(X_t)$. Then from \eqref{eqn:P0_growth} by integrating out the $y$ component we see that
    \[g(\theta^\phi;\P^0) = \frac18\int_E \div A^\top A^{-1} \div A - \frac12\int_E(\frac12 A^{-1} \div A - \nabla \phi)^\top A (\frac12 A^{-1} \div A - \nabla \phi). \]
    However, by the growth rate invariance property we have for any measure $\P \in \Pi_{\emptyset}$ -- and in particular for $\P^0$ -- that the growth rate of $\theta^\phi$ is given by \eqref{eqn:c2_growth_rate}. Equating the two expressions for the growth rate shows that minimizing the integral on the right hand side of \eqref{eqn:argmin} over $\phi \in \Dcal$ is equivalent to minimizing the integral in \eqref{eqn:unconstrained_min_problem} over $\phi \in \Dcal$. Since $\hat \phi \in \Dcal \subset W^{1,2}_{\mathrm{loc}}(E)$ is a minimizer of \eqref{eqn:unconstrained_min_problem} it follows that $\hat \phi$ also satisfies \eqref{eqn:argmin}. This completes the proof.
    	\end{proof}
    	
    We now work towards proving Lemma~\ref{lem:hat_v} and Theorem~\ref{thm:worst_case_measure_existence}. To this end fix a compact set $K \subset E$ and choose an open set $V$ such that $K \subset V \Subset E$. Let $\eta \in C_c^\infty(E)$ be nonnegative and such that $\eta = 1$ on $K$ and $\eta = 0$ on $E \setminus V$. Define the $K$-modification
    \begin{equation} \label{eqn:K-mod}
        \tilde c_X(x,y) = \eta(x)c_X(x,y) + \frac{1-\eta(x)}{p(x,y)}\frac{A(x)}{|D|}
    \end{equation}
    and set $f^x(y) = -\div_x(\tilde c_X(\tilde \ell_X - \nabla \hat \phi)p)(x,y)$ for $(x,y) \in F$. We are now ready to prove Lemma~\ref{lem:hat_v}.
    \begin{proof}[Proof of Lemma~\ref{lem:hat_v}]
    Fix $x \in E$. Since $\tilde c_X \in \Ccal_K$ and $\hat \phi$ satisfies \eqref{eqn:phi_pde} we have that
    \[\int_D f^x = - \int_D (\div_x(\tilde c_X(\tilde \ell_X - \nabla \hat \phi)p)(x,y)\, dy = -\div(\frac{1}{2}\div A(x) - A(x) \nabla \hat \phi(x)) = 0,\] so that $f^x_D = 0$.
    Moreover, by Assumption~\ref{ass:nondegeneracy}\ref{item:assum_new_3}, together with the form of the $K$-modification \eqref{eqn:K-mod} we have that $f^x/\sqrt{\lambda_{\mathrm{min}}(c_Y(x,\cdot))p(x,\cdot)}\in L^2(D)$ (the precise values of $C,b$ and $M$ used to apply Assumption~\ref{ass:nondegeneracy}\ref{item:assum_new_3} can be computed by expanding out the divergence term in the numerator and using the fact that $A,\nabla \hat \phi$ and $\eta$ do not depend on $y$). Hence, we can apply Theorem~\ref{T_pde_no_param} (with $a = c_Yp(x,\cdot)$, $\xi =  \ell_Y(x,\cdot)$ and $f= f^x$) to obtain a $\hat v(x,\cdot) \in W^{1,2}_{\mathrm{loc}}(D)$ satisfying \eqref{PDE_weak}. Thanks to Theorem~\ref{T_measurable} we may choose $(x,y) \mapsto \hat v(x,y)$ jointly Borel measurable, while preserving the above properties on a full-measure subset of $x \in E$. Next, since $\psi(x,\cdot) \in C_c^1(D)$ whenever $\psi \in C_c^1(F)$ we obtain \eqref{eqn:v_weak_pde} by integrating \eqref{PDE_weak} over $x \in E$, which establishes \ref{item:hat_v_pde}. 
    
    Theorem~\ref{T_pde_no_param} additionally yields that $\hat v(x,\cdot)$ satisfies the bound \eqref{PDE_norm_bound_1} for every $x \in E$. By squaring both sides, integrating over $x \in E$ and applying elementary bounds we obtain the first inequality in \eqref{eqn:y_l2_bound}. Consequently \ref{item:hat_v_bound} will follow once we show finiteness of the right hand side. By Assumption~\ref{ass:nondegeneracy}\ref{item:technical_1} we have that $\int_F \ell_Y^\top c_Y \ell_Y\, p < \infty$, so it just remains to prove that
    \[\int_F \frac{(\div_x(\tilde c_X (\tilde \ell_X - \nabla \hat \phi)p))^2}{\lambda_{\min}(c_Y)p} < \infty.\]
    To this end note that when $x \in E \setminus V$ then $\tilde c_X(x,y) = A(x)/(|p(x,y)|D|)$. Using the fact that $\hat \phi$ satisfies \eqref{eqn:phi_pde} we have whenever $x \in E \setminus V$ that
    \[\div_x(\tilde c_X \tilde \ell_X p)(x,y)= \div(\frac{1}{2}\div A(x))/|D| = \div(A(x)\nabla \hat \phi(x))/|D| = \div_x(\tilde c_X\nabla \hat \phi\, p)(x,y).\]
    Consequently, the numerator of the integrand is identically zero in this set so that
    \[\int_F \frac{(\div_x(\tilde c_X (\tilde \ell_X - \nabla \hat \phi)p))^2}{\lambda_{\min}(c_Y)p} = \int_V \int_D \frac{(\div_x(\tilde c_X (\tilde \ell_X - \nabla \hat \phi)p))^2}{\lambda_{\min}(c_Y)p}(x,y)\, dy\, dx.\]
    Again by Assumption~\ref{ass:nondegeneracy}\ref{item:hat_v_bound} we have that the inner integral is finite for every $x$, so the finiteness of the double integral follows from the continuity in $x$ of all the ingredients (i.e.\ continuity of $c_Y$ and continuity of derivatives of $\tilde c_X, p, \nabla \hat \phi,\eta$) on $\overline V$, which establishes \ref{item:hat_v_bound}.
    
It now just remains to prove that $\nabla_y \hat v \in L^q_{\mathrm{loc}}(F;\R^m)$ for every $q \in [2,\infty)$. Fix such a $q$. The main tool for proving this will be Theorem~\ref{T_regularity_eq}, but we first need to ensure that that the constant $C$ appearing in \eqref{eqn:lp_loc_bound} is independent of $x$. To this end let $U \Subset F$ be given and note that $U \subset U_D \times U_E$ for some $U_D \Subset D$ and $U_E \Subset E$. Pick $\varepsilon \in (0,\inf_{y \in U_D}\delta_{\partial U}(y))$ and set $U_{D,\varepsilon} := \{y \in \R^d: \delta_{U_D}(y) < \varepsilon\}$. Note that $U_{D,\varepsilon} \Subset D$. Next, since $c_Yp$ is uniformly elliptic on $U_\varepsilon \times U_{D,\varepsilon}$ and,  by Assumption~\ref{ass:regularity}\ref{item:regular_3} and Assumption~\ref{ass:nondegeneracy}\ref{item:c_y_sobolev}, $\div_y(c_Yp)$ is bounded on $U_\varepsilon \times U_{D,\varepsilon}$ we can choose a $\kappa \in (0,1)$ such that 
	\[\inf_{x \in U_E}\inf_{y \in U_{D,\varepsilon}} \lambda_{\min}(c_Yp(x,y)) \geq \kappa, \quad \|c_Yp\|_{L^\infty(U_E \times U_{D,\varepsilon})} \leq \kappa^{-1}, \quad \|\div_y (c_Yp)\|_{L^\infty(U_E \times U_{D,\varepsilon})} \leq \kappa^{-1}.
	\]
	Theorem~\ref{T_regularity_eq} now yields 
	\begin{equation} \label{eqn:loc_int_bnd}
	\begin{split}\|\hat v(x,\cdot)\|_{W^{2,q}(U_D)} \leq & C\Biggl(\|f^x\|_{L^q(U_{D,\varepsilon})} + \|\div_y(c_Y\ell_Y p)(x,\cdot)\|_{L^q(U_{D,\varepsilon})} \\
	& \hspace{0.5cm}+  \left\| \frac{f^x}{\sqrt{\lambda_{\mathrm{min}}(c_Y(x,\cdot))p(x,\cdot)}}\right\|_{L^2(D)} + \left(\int_D \ell_Y^\top c_Y \ell_Yp(x,y)\, dy\right)^{1/2}\Biggr),
	\end{split}
	\end{equation} 
	where the constant $C > 0$ is independent of $x \in U_E$ by our choice of $\kappa$. Since $f^x, c_Y$ and $p$ are continuous in $x$, the functions $\ell_Y,c_Y,p$ are locally bounded and $\div_y(c_Y\ell_Yp) \in L^q(U_E \times U_{D,\varepsilon})$, courtesy of Assumption~\ref{ass:regularity}\ref{item:regular_3} and Assumption~\ref{ass:nondegeneracy}\ref{item:c_y_sobolev}, we can raise both sides of \eqref{eqn:loc_int_bnd} to the power $q$ and integrate over $x \in U_E$ to establish that $\int_{U_E \times U_D} |\nabla_y \hat v(x,y)|^q \, dx \, dy < \infty$. Since $U \subset U_E\times U_D \Subset F$ was arbitrary we see that $\nabla_y \hat v \in L^q_{\mathrm{loc}}(F;\R^m)$. This completes the proof.     \end{proof}
We now turn our attention to proving Theorem~\ref{thm:worst_case_measure_existence}. The proof is broken up into several lemmas. Our method of construction uses the theory of \emph{generalized Dirichlet forms}. We refer to \cite{lee2020analyt} for terminology used below that is not explicitly defined in this paper. We start by introducing the relevant objects. It will be convenient to set $\tilde c = \mathrm{diag}(\tilde c_X, c_Y)$ and $\tilde \ell = (\tilde \ell_X, \ell_Y)$. We define the symmetric Dirichlet form $(\widetilde \Ecal^0, D(\widetilde \Ecal^0))$ as the closure on $L^2(F,\mu)$ of   
\begin{equation}\label{eqn:Ecal0}
\widetilde \Ecal^0(u,v) = \int_F \nabla u^\top \tilde c\, \nabla v\, p \qquad u,v \in C_c^\infty(F)
\end{equation}
and its corresponding generator $(\widetilde L^0,D(\widetilde L^0))$, which when acting on functions $ u \in C_c^\infty(F)$ has the form
\[\widetilde L^0 u  = \frac{1}{2}\Tr(\tilde c\, \nabla^2u) + \tilde \ell\,^\top \tilde c\,  \nabla u.\]
These objects are as in \eqref{eqn:symmetric_dirichlet_form_def} and \eqref{eqn:generator_reversible} with $c$ replaced by $\tilde c$. Next we define 
$\beta$ via
\begin{equation} \label{eqn:beta_def}
	\beta(x,y) = \tilde c(x,y)\left( \begin{pmatrix}
		\nabla \hat \phi(x) \\ \nabla_y \hat v(x,y) 	\end{pmatrix} - \tilde \ell(x,y)\right) = \begin{pmatrix}
		\tilde c_X(x,y)(\nabla\hat \phi(x) - \tilde \ell_X(x,y)) \\
		c_Y(x,y)(\nabla_y\hat v(x,y) - \ell_Y(x,y) )
	\end{pmatrix},
\end{equation}
for a.e.\ $(x,y) \in F$.
The first course of business is to obtain the existence and regularity of a certain semigroup.
\begin{lem} \label{lem:technical_lemma}
There exists an operator $(\widehat L^K, D(\widehat L^K))$ on $L^1(F,\mu)$ such that the following hold:
\begin{enumerate}
	\item $C_c^\infty(F) \subset D(\widehat L^K)$ and 
\begin{equation} \label{eqn:widehat_L_def}
	\widehat L^K u = \widetilde L^0u + \beta^\top \nabla u 
	;\qquad u \in C_c^\infty(F).
\end{equation}
\item \label{item:generalized_dirichlet_form}
 For every bounded $u \in D(\widehat L^K)$ and every compactly supported and bounded $v \in W^{1,2}(F)$ we have
\begin{equation} \label{eqn:generalized_ibp}
	\widetilde \Ecal^0(u,v) - \int_F v\beta^\top \nabla u \, p = -\int_F v\widehat L^Ku\, p.
\end{equation}
\item $(\widehat L^K, D(\widehat L^K))$ generates a strongly continuous contraction semigroup $(\widehat T_t^K)_{t \geq 0}$ on $L^1(F,\mu)$. Moreover $\widehat T_t^Kf$ has a continuous version $\widehat P_t^Kf$ for every $f \in \Bcal_b(F)$ and $t > 0$. \label{item:semigroup}
\end{enumerate}
\end{lem}
\begin{remark}
	In view of Lemma~\ref{lem:technical_lemma}\ref{item:semigroup} we will use the version $\widehat P^K_t$ of $\widehat T^K_t$ in the sequel.
\end{remark}
\begin{proof}
To apply the results of \cite{lee2020analyt} and \cite{stannat1999dirichlet} we need to verify that $\beta \in L^q_{\mathrm{loc}}(F;\R^{d+m})$ for some $q > d+m$ and 
\begin{equation} \label{eqn:beta_condition}
\int_F (\widetilde L^0u + \beta^\top \nabla u)p = 0
\end{equation} for every $u \in C_c^\infty(F)$.
Since $\widetilde L^0u\, p = \div(\tilde c\, \nabla u p)$ the divergence theorem yields $\int_F \widetilde L^0u\, p = 0$ for every $u \in C_c^\infty(F)$. Next note that by the local boundedness of $\tilde c,\tilde \ell$ and $\nabla \hat \phi$ together with Lemma~\ref{lem:hat_v}\ref{item:hat_v_regularity} we see that $\beta \in L^q_{\text{loc}}(F;\R^{d+m})$ for any $q \in [2,\infty)$;  in particular for $q > d+m$. Moreover, we have by the divergence theorem that for any $u \in C_c^\infty(F)$, 
\begin{equation} \label{eqn:beta_divergence_free}
	\begin{split}
		\int_F \beta^\top \nabla u\, p & = - \int_F \div (\beta p) u  \\
		& = \int_E \int_D \left(\div_x(\tilde c_X(\tilde \ell_X - \nabla \hat\phi)p) + \div_y (c_Y(\ell_Y - \nabla_y \hat v)p) \right)u(x,y)\,dy\,dx \\
		&= \int_E \int_D \left(\div_x(\tilde c_X(\tilde \ell_X - \nabla \hat\phi)p)u(x,y) - (\ell_Y - \nabla_y \hat v)^\top c_Y\nabla_yu\, p(x,y)\right)dy\,dx  = 0,
	\end{split}
\end{equation}
where the last equality follows from \eqref{eqn:v_weak_pde}. The first two items of the lemma now follow from \cite[Theorem~1.5]{stannat1999dirichlet}, while the last item follows from \cite[Theorem~2.31]{lee2020analyt}, which is applicable courtesy of \cite[Remark~2.40]{lee2020analyt}.
\end{proof}

Next we establish existence of the process. To construct the process we initially augment the state space $F$ with a cemetery state $\Delta$. To this end let $F_\Delta = F \cup \{\Delta\}$ be the one-point compactification of $F$. We set 
\[\Omega_\Delta := \{\omega \in C([0,\infty), F_\Delta): \omega_{t+h} = \Delta \text{ if } \omega_t = \Delta \text{ for all } h,t \geq 0\},\]
let $\Fcal_{\Delta}$ be the Borel $\sigma$-algebra induced by the topology of locally uniform convergence and (with a slight abuse of notation) denote the coordinate process by $Z$.
\begin{lem} \label{lem:process_existence}
	There exists a diffusion process \[\mathbb{M} =  (\Omega_\Delta,\F_\Delta,(\Fcal_t)_{t \geq 0}, (Z_t)_{t \geq 0},(\widehat \P^K_z)_{z \in F_\Delta})\] with state space $F$, lifetime 
	$\zeta := \inf\{t\geq0: Z_t = \Delta\}$ and transition semigroup $(\widehat P^K_t)_{t \geq0}$ given by Lemma~\ref{lem:technical_lemma}\ref{item:semigroup}. That is, for every $t \geq 0$, $z \in F$ and $f \in \Bcal_b(F)$ it holds that $\widehat P_t^Kf(z) = \widehat \E^K_z[f(Z_t)]$ where $\widehat \E^K_z[\cdot]$ denotes expectation under $\widehat \P^K_z$.  
\end{lem}
\begin{proof}
	By \cite[Theorem~3.5 and Proposition~3.6]{stannat1999dirichlet} together with \cite[Theorem~6]{trutnau2005hunt} we obtain a diffusion process $ \widecheck {\mathbb{M}} = (\Omega_\Delta, \F_\Delta, (\widecheck \F_t)_{t \geq 0 }, (\widecheck Z_t)_{t \geq 0}, (\widecheck \P^K_z)_{z \in F_\Delta})$ with state space $F$, lifetime $\widecheck \zeta := \inf\{t \geq 0: \widecheck Z_t = \Delta\}$ and such that for every $t \geq 0$ and $f \in \Bcal_b(F)$ we have
	\[\widehat P_t^Kf(z) = \widecheck\E^K_z[f(\widecheck Z_t)]; \quad \text{ for a.e. } z \in F.\]
	Using the (more than) strong Feller properties of the semigroup $(\widehat P^K_t)_{t \geq 0}$ developed in \cite[Section~2.3]{lee2020analyt} (which hold in this setting due to \cite[Remark~2.40]{lee2020analyt}) and following the proofs of \cite[Section~3]{lee2020analyt} verbatim up to and including Theorem~3.11 (which only depend on the results of Section~2.3 and 2.4) we obtain a diffusion process  $\mathbb{M} = (\Omega_\Delta, \F_\Delta, (\F_t)_{t \geq 0 }, ( Z_t)_{t \geq 0}, (\widehat \P^K_z)_{z \in F_\Delta})$ as in the statement of the lemma.
\end{proof}

\begin{lem} \label{lem:process_properties}
	The process $\mathbb{M}$ of Lemma~\ref{lem:process_existence} is strictly irreducible, recurrent and nonexplosive (i.e.\ $\zeta = \infty$, $\widehat \P^K_z$-a.s.\ for every $z \in F$). Moreover, $\mathbb{M}$ is a weak solution to \eqref{eqn:worst_case_dynamics}.
\end{lem}
\begin{proof}
The strict irreducibility is due to \cite[Proposition~2.39]{lee2020analyt}. To prove recurrence we use the criteria developed in \cite{gim2018recurrence} together with Assumption~\ref{ass:nondegeneracy}\ref{item:test_func}. First note that the same calculation as in \eqref{eqn:recurrence} yields that $\lim_{n \to \infty} \widetilde \Ecal^0(\chi_n,\chi_n) = 0$ where $\chi_n(x,y) = \varphi_n(x)\psi_n(y)$ and $\varphi_n,\psi_n$ are given by Assumption~\ref{ass:nondegeneracy}\ref{item:test_func}.
 Next note by Cauchy--Schwarz that

\begin{equation} \label{eqn:recurrence_bound}
	\begin{split} 
	\left(\int_{F}  \left|\beta^\top \nabla \chi_n\right| \, p\right)^2 & \leq\left(\int_F \nabla \chi_n^\top \tilde c\,  \nabla \chi_n\, p\right) \left(\int_F \left(\tilde \ell - \begin{pmatrix}
	\nabla_x \hat \phi \\ \nabla_y \hat v	\end{pmatrix}\right)^\top \tilde c\left(\tilde \ell - \begin{pmatrix}
	\nabla_x \hat \phi \\ \nabla_y \hat v 	\end{pmatrix}\right)\,p \right) \\
	& \hspace{-2.5cm} = \widetilde \Ecal^0(\chi_n,\chi_n)\left(\int_E(\frac{1}{2}A^{-1}\div A - \nabla \hat \phi)^\top A(\frac{1}{2}A^{-1}\div A - \nabla \hat \phi) - \frac{1}{4}\int_E \div A^\top A^{-1} \div A \right.\\
	& \left.\hspace{3cm}+ \int_F \tilde \ell_X\, ^\top \tilde c_X\, \tilde \ell_X\, p + \int_F (\ell_Y - \nabla_y \hat v)^\top c_Y(\ell_Y - \nabla \hat v)\, p\right).
\end{split} 
\end{equation} 
The equality followed from the fact that $\tilde c$ is block diagonal and by integrating out $y$ in the $\int_F(\tilde \ell_X - \nabla \hat \phi)^\top \tilde c_X (\tilde \ell - \nabla \hat \phi)p$ term. Since $\hat \phi$ is a minimizer for \eqref{eqn:unconstrained_min_problem}, the first integral on the right hand side is finite. As in the proof of Proposition~\ref{prop:hat_phi}, Assumption~\ref{ass:nondegeneracy}\ref{item:technical_1} implies that the second integral is finite. This, together with the definition of $\tilde c$, also implies that the third integral is finite. Finally, finiteness of the fourth integral follows from \eqref{eqn:y_l2_bound}. Consequently, the term on the left hand side of \eqref{eqn:recurrence_bound} tends to zero as $n \to \infty$. In summary, we have that
\[\lim_{n\to \infty} \widetilde \Ecal^0(\chi_n,\chi_n) + \int_F |\beta^\top \nabla \chi_n|\, p = 0. \]
\cite[Remark~15]{gim2018recurrence} together with \cite[Corollary~8(b)]{gim2018recurrence} now yield recurrence of the semigroup $(\widehat P^K_t)_{t > 0}$. Since, by Lemma~\ref{lem:process_existence}, $(\widehat P^K_t)_{t > 0}$ is the transition semigroup for the process $\mathbb{M}$ it immediately follows that the process is strictly irreducible recurrent.

 Next, note that by \cite[Corollary~20]{gim2018recurrence} $(\widehat P^K_t)_{t \geq 0}$ is conservative since it is recurrent, which yields the nonexplosiveness of $\mathbb{M}$ (see \cite[Corollary~3.23]{lee2020analyt}). The fact that $\mathbb{M}$ is a weak solution to \eqref{eqn:worst_case_dynamics} now follows via standard arguments by first connecting the process to the martingale problem for $\widehat L^K$ using \eqref{eqn:generalized_ibp} and then using the well-known equivalence between martingale problems and weak solutions to SDEs. Indeed following the proofs of \cite[Chapter~3]{lee2020analyt} verbatim from Proposition~3.12 onwards, but in our setting of a general open domain $F$ rather than all of $\R^{d+m}$, we obtain the result of Theorem~3.22(i), which establishes that $\mathbb{M}$ is a weak solution to \eqref{eqn:worst_case_dynamics}.
\end{proof}

We now establish the ergodicity of $\mathbb{M}$.
\begin{lem} \label{lem:ergodic}
	$\mu$ is an ergodic measure for $\mathbb{M}$ and \eqref{eqn:ergodic_property} holds for every locally bounded $h \in L^1(F,\mu)$. 
\end{lem}
\begin{proof}
	To establish that $\mu$ is an invariant measure we must show that for every $A \subset \Bcal(F)$ we have $\widehat P^K_t \mu(A) := \int_F \widehat P^K_t1_A(x)\, d\mu(x) = \mu(A)$. To this end define $\beta_* := - \beta$ and note that $\beta_* \in L^q_{\mathrm{loc}}(F;\R^{d+m})$ for any $q > d+m$ and it also satisfies \eqref{eqn:beta_condition}.
	Consequently, we obtain an operator $(\widehat L^K_*, D(\widehat L^K_*))$ and process $\mathbb{M}_*$ as in Lemmas~\ref{lem:technical_lemma} and \ref{lem:process_existence} respectively corresponding to $\beta_*$.
	Additionally the conclusions of Lemma~\ref{lem:process_properties} hold for $\mathbb{M}_*$ which, in particular, yield conservativity of the corresponding semigroup $(\widehat P^K_{*,t})_{t > 0}$. Moreover, it holds that $\widehat P^K_{*,t}$ is the adjoint operator of $\widehat P^K_t$ on $L^2(F,\mu)$ (see \cite[Remark~1.7]{stannat1999dirichlet} or \cite[Definition~2.7]{lee2020analyt}). Consequently by conservativity and the adjointness property we have for any $A \in \Bcal(F)$ and $t > 0$ that
	\[\widehat P_t^K\mu(A) = \int_F \widehat P^K_t1_A(x)\, d\mu(x) = \int_F 1_A(x)\widehat P^K_{*.t}1(x)\, d\mu(x) = \mu(A),\]
	establishing invariance.
	
	Define for $z \in F$ and $t >0$ the measures $\widehat P_t^K(z,\cdot):= \widehat \E^K_z[Z_t \in \cdot]$. Note that $(\widehat P^K_t)_{t \geq 0}$ is a stochastically continuous semigroup. Indeed, for any $z \in F$ and $r > 0$ such that $B_r(z) \subset F$ we have by right continuity that $\lim_{t \to 0}\widehat P^K_t(z,B_r(z)) = \lim_{t \to 0} \widehat \E^K_z[1_{B_r(z)}(Z_t)] = 1$. Hence by \cite[Proposition~4.1.1 and Theorem~4.2.1]{da1996ergodic} we see that $\mu$ is strongly-mixing and hence ergodic. Moreover, $\mu$ is the unique invariant measure for $(\widehat P^K_t)_{t \geq 0}$ and the measures $\widehat P^K_t(z,\cdot)$ are equivalent to $\mu$ for every $t > 0$ and $z \in F$. 
	
	Next we establish that $Z$ is a stationary process under the measure $\widehat \P^K_\mu$ given by $\widehat \P^K_\mu(A) = \int_F \widehat \P^K_z(A)\, d\mu(z)$ for $A \in \Fcal$. Indeed note that for $f \in C^2_0(F)$ we have 
	\[\widehat \E^K_\mu[f(Z_t)] = \widehat \E^K_\mu[f(Z_0)] + \int_0^t \widehat \E^K_\mu[\widehat L^Kf(Z_s)]\, ds\]  and, again by conservativity and adjointness, that 
	\[\widehat \E^K_\mu[\widehat L^Kf(Z_s)] = \int_F \widehat P^K_s\widehat L^Kf(z)\, d\mu(z) = \int_F \widehat L^Kf(z)\, d\mu(z) = 0, \]
	where the last equality followed by \eqref{eqn:widehat_L_def} and \eqref{eqn:beta_condition}. By approximation we now obtain that $\widehat \E^K_\mu[f(Z_t)] = \widehat \E^K_\mu[f(Z_0)]$ for every $f \in C_b(F)$ so that, under $\widehat \P^K_\mu$, $\mathrm{Law}(Z_t) = \mu$ for every $t \geq 0$. As a consequence we obtain (see e.g.\ \cite[Corollary~25.9]{kallen2021found}) for any measurable nonnegative $h$ that
	\begin{equation} \label{eqn:ergodic_prop_mu}
	\lim_{T \to \infty} \frac{1}{T}\int_0^T h(Z_t)\, dt = \int_F hp; \quad \widehat \P^K_\mu\text{-a.s.}
	\end{equation}
	By definition of $\widehat \P^K_\mu$ we see that the convergence in \eqref{eqn:ergodic_prop_mu} also holds $\widehat \P^K_z$-a.s.\ for almost every $z \in F$. 
	
	To obtain the ergodic property for \emph{every} $z \in F$ we fix a locally bounded $h \in L^1(F,\mu)$ and argue as in \cite[Theorem~4.7.3(iv)]{Fukushima2011Dirichlet}. We reproduce the proof of that result for the reader's convenience. Let 
	\[\Lambda = \left\{\omega \in \Omega: \lim_{T \to \infty} \frac{1}{T}\int_0^T h(Z_t(\omega))\, dt = \int_F hp\right\}.\] We have already deduced that $\widehat \P^K_z(\Lambda) = 1$ for $z \in F\setminus N$, where $N$ is some Lebesgue null set. Next define $\Gamma_n = \{\omega \in \Omega: \int_0^{1/n} h(Z_t(\omega))\, dt < \infty \}$.
	For a fixed $z \in F$, it follows that $\lim_{n \to \infty} \widehat \P^K_z(\Gamma_n) = 1$ by the right continuity of $Z_t$ and the local boundedness assumption on $h$. Moreover, it is clear that $\theta^{-1}_{1/n}\Lambda \cap \Gamma_n \subset \Lambda$ where $\theta_t:\Omega \to \Omega$ for $t > 0$ is the shift operator: $\theta_t(\omega(\cdot)) = \omega(\cdot + t)$. 
	
	Since $\widehat \P^K_z(Z_{1/n} \in N) = 0$ by the equivalence of $\widehat P^K_{1/n}(z,\cdot)$ with $\mu$ we see that 
	\[\widehat \P^K_z(\Lambda) \geq \widehat \P^K_z(\theta^{-1}_{1/n}\Lambda \cap \Gamma_n) = \widehat \E^K_z[\widehat \P^K_{Z_{1/n}}(\Lambda); \Gamma_n, Z_{1/n} \in F\setminus N] = \widehat\P^K_z(\Gamma_n).\]
	Sending $n \to \infty$ now yields that $\widehat\P^K_z(\Lambda) = 1$, which completes the proof. 
\end{proof}

\begin{proof}[Proof of Theorem~\ref{thm:worst_case_measure_existence}]
The existence of the SDE follows from Lemma~\ref{lem:process_properties} and the ergodic property \eqref{eqn:ergodic_property} follows from Lemma~\ref{lem:ergodic}. The tightness of the laws of $\{X_t\}_{t > 0}$ under $\widehat \P^K_{(x,y)}$ for every $(x,y) \in F$ are clear from the ergodic property of $\mu$ (see e.g.\  \cite[Theorem~4.2.1(i)]{da1996ergodic}) and establishes that $\widehat \P^K_{(x,y)} \in \Pi_K$.
\end{proof}

Next we prove Corollary~\ref{cor:original}.
\begin{proof}[Proof of Corollary~\ref{cor:original}] The result will follow in the same way as Theorem~\ref{thm:main} if we can construct a worst-case measure $\widehat \P \in \Pi_0$ under which $\hat \theta$ is growth-optimal and $c_X$ is the instantaneous covariation matrix for $X$. Since $A^{-1} \div A$ is a gradient of a function we see from the variational problem \eqref{eqn:unconstrained_min_problem} that $\nabla \hat \phi = \frac{1}{2} A^{-1} \div A$. Hence \eqref{eqn:grad_int_assumption} is equivalent to \eqref{eqn:ideal_bound}.

 Next we obtain a $\hat v$ satisfying the three items in Lemma~\ref{lem:hat_v} with $c_X$ replacing $\tilde c_X$. In particular, as a consequence of \eqref{eqn:ideal_bound}, we have finiteness in \eqref{eqn:y_l2_bound}. We now construct the operator $(\widehat L, D(\widehat L))$ and corresponding semigroup $\widehat P$ as in Lemma~\ref{lem:technical_lemma} with $c_X$ replacing $\tilde c_X$. We obtain a process $\mathbb{M}$ as in Lemma~\ref{lem:process_existence} corresponding to $\widehat P$. The properties of Lemma~\ref{lem:process_properties} hold for this $\mathbb{M}$ as well -- in particular we are able to get recurrence without needing a $K$-modification courtesy of \eqref{eqn:ideal_bound}. The ergodic property then follows as in Lemma~\ref{lem:ergodic}, which yields a measure $\widehat \P_{(x,y)}\in \Pi_0$ for every $(x,y) \in F$, where the dynamics of $(X,Y)$ under $\widehat\P_{(x,y)}$ are given by \eqref{eqn:worst_case_dynamics} with $c_X$ replacing $\tilde c_X$. This completes the proof.
\end{proof}

\subsection{Proofs for Section~\ref{sec:main_extended}} \label{sec:proofs_extended}
Lastly, we turn towards establishing Theorem~\ref{thm:main_extended}. Akin to \eqref{eqn:K-mod} we start by defining an explicit $K$-modification $\tilde c$ of the input matrix $c$. To this end fix a compact set $K \subset F$  and choose $K_E \subset E$ and $K_D \subset D$ compact such that $K \subset K_E \times K_D \subset F$. Let $U$ and $V$ be open sets such that $K_E \subset V \Subset E$ and $K_D \subset U \Subset D$. Next we choose $\eta_E \in C_c^\infty(E)$ and $\eta_D \in C_c^\infty(D)$ which are nonnegative and satisfy $\eta_E = 1$ on $K_E$, $\eta_E = 0$ on $E\setminus V$, $\eta_D = 1$ on $K_D$ and $\eta_D = 0$ on $D \setminus U$. Finally we define a $K$-modification of $c$ via
\begin{equation} \label{eqn:K-mod-extended}
\tilde c(x,y) := \begin{bmatrix}
\eta_E(x)c_X(x,y) + \frac{1-\eta_E(x)}{p(x,y)}\frac{A(x)}{|D|} &  \eta_D(y)c_{XY}(x,y) \\
\eta_D(y)c_{XY}^\top (x,y) & c_Y(x,y).
\end{bmatrix}
\end{equation}
It is with this matrix that we will construct a solution $\widehat \P^K_c$ to \eqref{eqn:worst_case_extended}. The first order of business is to prove Lemma~\ref{lem:hat_vc}.

\begin{proof}[Proof of Lemma~\ref{lem:hat_vc}]
To establish the result we need to show that $\int_D g^x(y)\, dy = 0$ for every $x \in E$ where
\[g^x(y) = -\mathrm{div}_x(\tilde c_X(\tilde \xi_X - \nabla \hat \phi)p)(x,y).\]
Indeed, once this is established the remainder of the proof follows in exactly the same way as the proof of Lemma~\ref{lem:hat_v} with $g^x$ replacing $f^x$. The integrability conditions of Assumption~\ref{ass:nondegeneracy_joint} replace those of Assumption~\ref{ass:nondegeneracy} to ensure that the required estimates that held for $f^x$ in the proof of Lemma~\ref{lem:hat_v} hold for $g^x$ here.

Now, to evaluate the integral of $g^x$ we first compute $\tilde \xi$ in terms of $\tilde c$ and $\tilde \ell$,
\begin{align*}
	\tilde \xi & = \begin{bmatrix} \tilde c_X^{-1} & 0 \\
	0 & \tilde c_Y^{-1}
	\end{bmatrix} \begin{bmatrix} \tilde c_X & \tilde c_{XY} \\
	\tilde c_{XY}^\top & \tilde c_Y \end{bmatrix} \tilde \ell\\
&  = \begin{bmatrix} I_d & \tilde c_X^{-1}\tilde c_{XY} \\ \tilde c_Y^{-1}\tilde c_{XY}^\top   & I_m\end{bmatrix} \tilde \ell  \\
&  = \begin{bmatrix} \tilde \ell_X + \tilde c_X^{-1}\tilde c_{XY}\tilde \ell_Y  \\
	\tilde c_Y^{-1}\tilde c_{XY}^\top \tilde  \ell_X + \tilde \ell_Y
	\end{bmatrix}. 
\end{align*}
Consequently, we see that
 \begin{align}
	g^x & = - \mathrm{div}_x(\tilde c_X(\tilde \ell_X - \nabla \hat \phi)p) - \mathrm{div}_x(\tilde c_{XY}\tilde \ell_Y p) \nonumber \\ 
	& = - \mathrm{div}_x(\tilde c_X(\tilde \ell^0_X - \nabla \hat \phi)p) - \mathrm{div}_x(\tilde c_{XY}\tilde \ell_Y p) - \mathrm{div}_x(\tilde c_X(\tilde \ell_X - \tilde \ell^0_X)p), \label{eqn:gx}
\end{align} 
where $\tilde \ell^0 = \frac{1}{2}(\tilde c^0)^{-1}\mathrm{div}\tilde c^0 + \frac{1}{2}\nabla \log p$. Since $\tilde c^0$ is block diagonal we see that $\tilde \ell^0_X = \frac{1}{2}\tilde c_X^{-1}\mathrm{div}\tilde c_X + \frac{1}{2}\nabla_x \log p$. Hence the first term on the right hand side coincides with $f^x$ and we have 
\[-\int_D \mathrm{div}_x(\tilde c_X(\tilde \ell_X^0 - \nabla \hat \phi)p)\, dy = \int_D f^x = 0.\]
Thus we just need to establish that the remaining terms integrate to zero. To this end it will be useful to define $B = \tilde c p$ and $B^0 = \tilde c^0 p$ so that the remaining terms on the right hand side of \eqref{eqn:gx} are given by the $x$-divergence of
\begin{align*} 
-B_{XY}((B^{-1}\mathrm{div} B)_Y) -B_X((B^{-1}\mathrm{div}B)_X) + \mathrm{div}(B^0_X) & = -(BB^{-1}\mathrm{div} B)_X + \mathrm{div}(B_X) \\
& = \mathrm{div}(B_X) - (\mathrm{div B})_X,
\end{align*}
where we used the fact that $B^0_X = B_X$.
Now fixing $i = 1,\dots,d$ we see that 
\[\mathrm{div}(B_X)^i -(\mathrm{div} B)^i = \sum_{j=1}^d \partial_{x_j} B_X^{ij} - \sum_{j=1}^m \partial_{y_j} B_{XY}^{ij} - \sum_{j=1}^d \partial_{x_j} B_{X}^{ij} = -\mathrm{div}_y( B_{XY}^i). \]
Hence (by interchanging $x$ and $y$ derivatives) we have that 
\[- \mathrm{div}_x(\tilde c_{XY}\tilde \ell_Y p) - \mathrm{div}_x(\tilde c_X(\tilde \ell_X - \tilde \ell^0_X)p) = - \mathrm{div}_y(\mathrm{div}_x( \tilde c_{XY}^\top p)).\]
Since this term is a $y$-divergence and $\tilde c_{XY} = 0$ on $D \setminus U$ we see by the divergence theorem that
\[\int_D (\mathrm{div}_x(\tilde c_{XY}\tilde \ell_Y p)(x,y) + \mathrm{div}_x(\tilde c_X(\tilde \ell_X - \tilde \ell^0_X)p)(x,y))\, dy = \int_D \mathrm{div}_y(\div_x(\tilde c_{XY}^\top p))(x,y)\, dy = 0,\]
which establishes that $\int_D g^x = 0$ and completes the proof.
\end{proof}

With Lemma~\ref{lem:hat_vc} established, the construction of $\widehat \P^K_c$ follows the same road map as the proof of Theorem~\ref{thm:main_extended}. As before we define the symmetric Dirichlet form $\widetilde \Ecal^0$ via \eqref{eqn:Ecal0}. Here, however, due to the different structure of the drift we do not take $\beta$ given by \eqref{eqn:beta_def} to determine the perturbation of the Dirichlet form, but rather define
\[\beta_{\tilde c}(x,y) := \tilde c^0(x,y)\begin{pmatrix}
    \nabla \hat \phi(x) \\\ \nabla_y \hat v_c(x,y)
\end{pmatrix} - \tilde c(x,y) \tilde \ell(x,y) = \begin{pmatrix}
    \tilde c_X(x,y)(\nabla \hat \phi(x) - \tilde \xi_X(x,y)) \\ \tilde c_Y(x,y)(\nabla_y \hat v_c(x,y) - \tilde \xi_Y(x,y))
\end{pmatrix}. \]
Then Lemma~\ref{lem:technical_lemma} holds with $\beta_{\tilde c}$ in place of $\beta$. Indeed, the key integration by parts calculation akin to \eqref{eqn:beta_divergence_free} yields for any $u \in C_c^\infty(F)$ that 
\begin{equation*} 
	\begin{split}
		\int_F \beta_{\tilde c}^\top \nabla u\, p & = - \int_F \div (\beta_{\tilde c} p) u  \\
		& = \int_E \int_D \left(\div_x(\tilde c_X(\tilde \xi_X - \nabla \hat\phi)p) + \div_y (c_Y(\tilde \xi_Y - \nabla_y \hat v_c)p) \right)u(x,y)\,dy\,dx \\
		&= \int_E \int_D \left(\div_x(\tilde c_X(\tilde \xi_X - \nabla \hat\phi)p)u(x,y) - (\tilde \xi_Y - \nabla_y \hat v_c)^\top c_Y\nabla_yu\, p(x,y)\right)dy\,dx  = 0,
	\end{split}
\end{equation*}
where the final equality follows since $\hat v_c$ is weak solution to \eqref{eqn:hat_vc_weak}. The construction of the corresponding diffusion as in Lemma~\ref{lem:process_existence}, the verification of its properties as in Lemma~\ref{lem:process_properties} and the ergodic property as in Lemma~\ref{lem:ergodic} then follow in exactly the same way as in the construction of $\widehat \P^K$. This establishes that the law $\widehat \P^K_c$ of \eqref{eqn:worst_case_extended} is in $\Pi_K^c$. In conjunction with the discussion of Section~\ref{sec:main_extended}, this establishes Theorem~\ref{thm:main_extended}.

\bibliographystyle{plain}
\bibliography{References2}

\end{document}